\documentclass[12pt]{article}
\usepackage[margin=1in]{geometry}
\usepackage{amsmath, amssymb, amsthm, bm}
\usepackage{hyperref}
\usepackage{natbib}
\usepackage{graphicx}
\usepackage{booktabs}
\usepackage{threeparttable}

\title{Who Matters to Whom? Identifying Peer Effects with Propagation Geometry\\}
\author{Guy Tchuente}
\date{February 2026}

\newcommand{\E}{\mathbb{E}}

\newcommand{\R}{\mathbb{R}}
\newcommand{\mc}{\mathcal}

\theoremstyle{plain}
\newtheorem{assumption}{Assumption}

\newtheorem{remark}{Remark}
\newtheorem{theorem}{Theorem}

\begin{document}
\maketitle

\begin{abstract}
This paper develops a unifying theory of peer effects that treats the peer aggregator (the social norm mapping peers' actions into a scalar exposure) as the central behavioral primitive.
We formulate peer influence as a norm game in which payoffs depend on own action and an exposure index, and we provide equilibrium existence and uniqueness for a broad class of aggregators. Using economically interpretable axioms, we organize commonly used exposure maps into a small taxonomy that nests
linear-in-means, CES (peer-preference) norms, and smooth ``attention-to-salient-peers'' aggregators; rank-based quantile norms are treated as a complementary class. Building on this unification, we show that each aggregator induces an operator that governs how exogenous variation propagates through the network. Linear-in-means corresponds to constant transport (adjacency matrix), recovering the classic (friends-of-friends) instrument families. For nonlinear norms, operator becomes state- and preference-dependent and is characterized by the Jacobian of the exposure map evaluated at an exogenous predictor. This perspective yields geometry-induced instrument that exploit heterogeneity in marginal influence and nonredundant paths, and can remain informative when one-step moments or adjacency-power instruments become weak. Monte Carlo evidence and an application to NetHealth illustrate the practical implications across
alternative aggregators and outcomes.
\end{abstract}



\newpage

\section{Introduction}

Peer effects are central to empirical work in education, health, crime, and technology adoption, where an individual's outcome responds to the outcomes or actions of peers connected through social, spatial, or organizational networks.
The econometric challenge is that peer exposure is endogenous: it is mechanically correlated with unobservables (the ``reflection problem'') and with shared shocks, so identification hinges on generating excluded variation that shifts peer exposure without directly shifting the outcome \citep{manski1993identification}.
A large empirical literature has therefore adopted the linear-in-means (LIM) model and its network generalizations, because LIM delivers transparent best-response foundations and a tractable class of instruments based on higher-order neighbor exposures such as $G^2X$ and $G^kX$ (see \citealp{bramoulle2009identification,blume2011identification,goldsmith2013social}).
However, two facts are increasingly hard to ignore.
First, the network topologies that arise in applications---high transitivity, community structure, and near-regular degree profiles---are precisely those in which adjacency-power instruments can become nearly collinear and weak.
Second, LIM hard-codes a behavioral assumption: the relevant peer benchmark is the \emph{mean} peer outcome, ruling out salience (``the best peer matters most''), downside comparison (``the weakest peer matters''), and rank-based norms (``keep up with the median/upper quartile'').

Recent work has begun to relax the mean benchmark by treating the peer aggregator (the ``social norm'') as an estimand.
\citet{boucher2024toward} (BRUZ) develop a general peer-effects model in which the relevant peer benchmark is a CES aggregator with curvature parameter $\beta$, nesting LIM as $\beta=1$ and converging to max/min norms as $\beta\to\pm\infty$.
In their framework, $\beta$ is economically meaningful---it governs \emph{who matters}---and identifying it is part of the empirical task.
They propose a one-step IV strategy based on scalar summaries of an exogenous predictor $\hat y=m(X)$ and derivatives of the aggregator with respect to $\beta$.

This paper provides a structural unification that nests (i) the LIM network model of \citet{bramoulle2009identification}, (ii) CES peer-preference norms as in \citet{boucher2024toward}, and (iii) ``attention to salient peers'' logic that underlies Brock--Durlauf-type social interactions \citep{brock2001discrete}.
Section~\ref{sec:norm_games} formulates peer effects as a \emph{norm game}: individuals choose actions, payoffs depend on own action and a scalar peer exposure index $N_i=\Phi_i(a_{-i};G,\theta)$, and equilibria are fixed points of best responses.
Rather than impose a single functional form, we treat the peer exposure map $\Phi_i$ as a primitive and discipline it using economically interpretable axioms.
Under mild regularity, the axioms yield the Kolmogorov--Nagumo class of quasi-arithmetic means, $\Phi_i(a)=\varphi^{-1}(\sum_j g_{ij}\varphi(a_j))$, providing a common lens on a large fraction of peer-exposure objects used in economics.
Within this lens, LIM corresponds to the linear generator $\varphi(a)=a$, CES corresponds to the power generator $\varphi(a)=a^\beta$, and a smooth-max ``attention'' norm arises from log-sum-exp aggregation, which also has classic discrete-choice and entropy-regularization interpretations \citep{mcfadden1974conditional,luce1959individual,matejka2015rational,nesterov2005smooth}.
The appendix adds a complementary class of rank-based norms via loss minimization (weighted quantiles/medians), which are natural for robustness and for Manski-style nonparametric arguments \citep{koenker1978quantiles,manski1993identification}.
Section~\ref{sec:norm_games} also provides equilibrium existence for this unified class (and uniqueness under transparent ``weak interaction'' conditions), so the peer-preference parameter $\theta$ is a genuine behavioral primitive rather than an ad hoc index.

A constructive identification toolkit follows naturally from the norm-game taxonomy.
The central insight is that \emph{propagation is governed by a transport operator induced by the peer aggregator}.
In LIM the transport is constant and equals adjacency ($P=G$), so instrument design reduces to repeated linear propagation ($G^kX$).
In CES and other smooth aggregators, the exposure map implies a \emph{state- and preference-dependent} Jacobian transport operator evaluated at an exogenous predictor $\hat y$:
\[
w_{ij}(\theta;\hat y)=\left.\frac{\partial \Phi_i(y;G,\theta)}{\partial y_j}\right|_{y=\hat y},
\qquad
P_{ij}(\theta;\hat y)=\frac{w_{ij}(\theta;\hat y)}{\sum_{m\neq i}w_{im}(\theta;\hat y)}.
\]
This operator encodes ``who matters to whom'' under the maintained aggregator, and it is the appropriate object for multi-step propagation when the model is not linear-in-means.
Section~\ref{sec:norm_to_np} shows that once $P(\theta;\hat y)$ is taken seriously as the propagation primitive, instrument construction becomes a geometry problem rather than an algebra-of-$G$ problem:
the analogue of $G^kX$ is $P^k(\theta;\hat y)X$, and the transport induces additional geometry objects---effective distances (shells/geodesics) and path nonredundancy (torsion)---that generate new excluded-variation directions.
These directions are informative precisely in settings where one-step scalar moments (BRUZ-type summaries) or adjacency-power instruments become weak due to transitivity, regularity, and path redundancy.

This paper relates to the econometrics of social interactions and the identification of peer effects, beginning with the reflection problem \citep{manski1993identification} and extending to network-based identification strategies such as \citet{bramoulle2009identification} and subsequent work on network instruments and their limitations (see also \citealp{blume2011identification,goldsmith2013social,boucherBramoulle2026binary}). Relatedly, \citet{tchuente2019weak} emphasizes that network-IV strategies may require many or highly correlated instruments,
leading to weak identification or many-instruments bias, and proposes regularized 2SLS estimators as a remedy.
It complements structural and reduced-form literatures that emphasize heterogeneity in ``which peers matter,'' including the CES peer-preference approach of \citet{boucher2024toward} and recent results showing that allowing richer action structures can fundamentally change shock propagation in network games \citep{sadlerYeh2026}.
It also connects to the discrete-choice and information-theoretic foundations of attention/salience norms \citep{luce1959individual,mcfadden1974conditional,matejka2015rational} and to quantile-based loss minimization \citep{koenker1978quantiles}, as well as to recent work defining network structure behaviorally via equilibrium implications \citep{jacksonStorms2026}.
More broadly, it speaks to the econometric analysis of network data and strategic interaction models \citep{graham2020econometric,kline2020econometric,rootSadler2026}, by providing a constructive route from a maintained aggregator to an instrument menu that matches the implied influence geometry.

Section~\ref{sec:norm_games} develops the norm-game framework and the taxonomy of peer aggregators, and establishes equilibrium existence (and uniqueness under weak interactions).
Section~\ref{sec:norm_to_np} develops geometry-induced instrument families implied by the aggregator-induced transport.
Section~\ref{sec:manski_to_geometry} discusses identification and connects completeness-type conditions to the enlarged excluded information set generated by transport geometry.
Section~\ref{sec:mc} provides Monte Carlo evidence.
Section~\ref{sec:nethealth} applies the framework to NetHealth and compares CES, smooth-max, and quantile norms across steps, sleep, and GPA.
Section~\ref{sec:conclusion} concludes and discusses extensions.
\section{Norm games and a unifying class of peer aggregators}
\label{sec:norm_games}

This section builds a unifying structural framework that nests:
(i) the linear-in-means network model of \citet{bramoulle2009identification},
(ii) generalized peer norms with peer-preference parameters as in \citet{boucher2024toward}, and
(iii) the ``attention to salient peers'' logic that underlies Brock--Durlauf-type social interactions
\citep{brock2001discrete}.
We isolate a small set of primitives, state economically interpretable aggregator classes,
and provide equilibrium existence (and uniqueness under transparent ``weak interaction'' conditions).
The appendix then adds an additional aggregator class (quantile/median norms) and additional technical proofs. 

\subsection{Primitives: actions, network, and a peer aggregator}
\label{subsec:primitives}

Individuals $i=1,\dots,n$ choose an action $a_i$ in a feasible set $A_i\subset \mathbb{R}$ (extensions to $\mathbb{R}^d$ are immediate).
Let $a=(a_1,\dots,a_n)$ and $a_{-i}$ denote the vector without component $i$.

The network is summarized by a row-stochastic interaction matrix $G=(g_{ij})$ with $g_{ij}\ge 0$ and $\sum_{j\neq i} g_{ij}=1$
(when observations are grouped, e.g.\ schools $s$, all objects can be indexed by $s$; we suppress $s$ for readability.)

\paragraph{Peer aggregator (norm/exposure).}
We model peer influence through a scalar peer exposure index
\begin{equation}\label{eq:peer_agg_general}
N_i(a_{-i}) \;=\; \Phi_i(a_{-i};G,\theta),
\end{equation}
where $\Phi_i$ is an aggregator mapping peers' actions to a scalar, and $\theta$ indexes \emph{peer preference/salience}
(e.g.\ a CES parameter $\beta$ or an attention parameter $\kappa$).

\paragraph{Payoffs.}
Preferences depend on own action and peer exposure:
\begin{equation}\label{eq:utility_general}
u_i(a_i,a_{-i})
\;=\;
U_i\!\big(a_i,\; N_i(a_{-i}),\; x_i\big),
\end{equation}
where $x_i$ are observables (and unobservables can be added additively).

\paragraph{Nash equilibrium.}
A pure-strategy Nash equilibrium is $a^\star\in \prod_i A_i$ such that, for each $i$,
\begin{equation}\label{eq:nash_def}
a_i^\star\in \arg\max_{a_i\in A_i} U_i\!\Big(a_i,\;\Phi_i(a^\star_{-i};G,\theta),\;x_i\Big).
\end{equation}

\subsection{A taxonomy of peer aggregators}
\label{subsec:taxonomy}

Rather than fix a single functional form, we treat the \emph{peer exposure map} as a primitive:
for each $i$, let
\[
N_i \;=\;\Phi_i(a_{-i};g_i),
\qquad g_i=(g_{ij})_{j\neq i},\ \ g_{ij}\ge 0,\ \ \sum_{j\neq i} g_{ij}=1,
\]
where $a_j$ is a peer action/outcome and $g_{ij}$ are interaction weights. We then discipline $\Phi_i$
using economically transparent \emph{axioms} that appear (explicitly or implicitly) across the peer-effects,
aggregation, and discrete-choice literatures.

\paragraph{Axioms that deliver a ``mean-like'' class (Kolmogorov--Nagumo).}
A large share of peer-exposure objects used in economics can be motivated as a notion of ``weighted mean''
satisfying four basic requirements:
\begin{enumerate}
  \item[(A1)] \textbf{Anonymity / symmetry (given weights):} relabeling peers leaves $N_i$ unchanged.
  \item[(A2)] \textbf{Monotonicity:} if $a_j$ increases (holding others fixed), then $N_i$ does not decrease.
  \item[(A3)] \textbf{Idempotence:} if all peers choose the same level $a$, then $\Phi_i(a,\dots,a)=a$.
  \item[(A4)] \textbf{Decomposability / associativity:} aggregating peers in two stages (first within subgroups,
  then across subgroups) gives the same result as aggregating all at once.
\end{enumerate}
Under mild regularity (continuity and strict monotonicity), (A1)--(A4) characterize the class of
\emph{quasi-arithmetic means}:
\begin{equation}\label{eq:quasi_arithmetic}
\Phi_i(a_{-i};g_i)\;=\;\varphi^{-1}\!\Big(\sum_{j\neq i} g_{ij}\,\varphi(a_j)\Big)
\quad\text{for some strictly increasing }\varphi,
\end{equation}
a classic result associated with Kolmogorov and Nagumo and presented in standard functional-equations treatments
(e.g.\ \citet{aczel1966lectures}).\footnote{Many modern summaries refer to this as the Kolmogorov--Nagumo theorem
for (weighted) quasi-arithmetic means; see, e.g., \citet{singpurwalla2020mean} for an accessible discussion.}
Equation \eqref{eq:quasi_arithmetic} provides a \emph{unifying lens}: different economic models correspond to different
choices of the generator $\varphi$ (and sometimes to additional structure restricting $\varphi$).

\paragraph{From ``mean-like'' to economically interpretable families.}
We now list four families that are both empirically useful and microfounded by simple strengthening of the axioms.

\begin{enumerate}
\item \textbf{Mean / linear-in-means (LIM).}
Add \emph{translation invariance} (adding a constant to all peers shifts $N_i$ by the same constant) and
\emph{homogeneity of degree one} (scaling all peers scales $N_i$).
Within \eqref{eq:quasi_arithmetic}, these restrictions pin down (up to affine normalization) $\varphi(a)=a$,
so $\Phi_i$ becomes the weighted arithmetic mean. This is the benchmark exposure used in spatial-lag and LIM peer-effects
models \citep{bramoulle2009identification}.

\item \textbf{Power mean / CES norm (peer preferences).}
Impose instead that peer actions are \emph{homothetic} and that the aggregator has \emph{constant elasticity of substitution}
across peers. This leads to the CES/power-mean family, with generator $\varphi(a)=a^\beta$ (after normalization):
\[
\Phi_i^{\text{CES}}(a_{-i})
\;=\;
\Big(\sum_{j\neq i} g_{ij}\,a_j^\beta\Big)^{1/\beta}.
\]
This is exactly the ``peer preference'' norm in \citet{boucher2024toward}, and it is the same functional family
popularized in economics by the CES tradition \citep{arrow1961capital, dixit1977monopolistic}. Economically,
$\beta$ governs \emph{concentration/salience}: $\beta>1$ makes high-$a_j$ peers disproportionately influential
(max-like), while $\beta<1$ makes low-$a_j$ peers disproportionately influential (min-like).

\item \textbf{Smooth-max (attention / extreme-peer salience).}
If one wants (i) \emph{extreme-peer salience} (the norm approaches $\max_j a_j$ as attention sharpens) and
(ii) \emph{smoothness} (differentiability in peers' actions, useful for comparative statics and estimation),
a canonical choice is the log-sum-exp (a smooth approximation to $\max$):
\[
\Phi_i^{\text{smax}}(a_{-i})
\;=\;
\kappa^{-1}\log\sum_{j\neq i} g_{ij}\exp(\kappa a_j),
\qquad \kappa>0.
\]
This object has two widely used economic interpretations.
First, it is the ``inclusive value'' in logit-type discrete choice models (random utility with i.i.d.\ Gumbel shocks),
grounded in classic choice axioms and the logit structure \citep{mcfadden1974conditional}.
Second, it appears as the value of an entropy-regularized max problem (a rational-inattention / Gibbs variational form),
connecting ``attention'' to Shannon entropy penalties \citep{matejka2015rational}.
From a mathematical perspective, log-sum-exp is also the standard smooth approximation of $\max$ used in convex optimization
(e.g.\ \citet{nesterov2005smooth}), which is convenient because it yields simple, interpretable marginal influence weights.

\item \textbf{Quantile / median norms (robust, nonparametric-friendly).}
For Manski-style nonparametric reasoning, it is often natural to replace ``average pressure'' by a \emph{status/rank target},
e.g.\ ``keep up with the median'' or ``stay above the $q$-th percentile.'' A clean axiom here is \emph{loss-minimization}:
define $N_i$ as the minimizer of an asymmetric absolute loss (the ``check loss''):
\begin{equation}\label{eq:quantile_norm}
\Phi_i^{(q)}(a_{-i})
\;:=\;
\arg\min_{t\in\mathbb{R}}
\sum_{j\neq i} g_{ij}\,\rho_q(a_j-t),
\qquad
\rho_q(u)=u\big(q-\mathbf{1}\{u<0\}\big),
\end{equation}
which yields the (weighted) $q$-quantile of peers' actions; $q=\tfrac12$ gives the weighted median.
This characterization is standard in the quantile regression literature \citep{koenker1978quantiles}.
Quantile norms are appealing in peer-effects settings because they are robust to outliers and naturally compatible with
nonparametric identification discussions in the spirit of \citet{manski1993identification}.
\end{enumerate}

\begin{table}[!htbp]\centering
\caption{A taxonomy of peer aggregators (peer exposure maps) and their axiomatic/microfounded motivation}
\label{tab:taxonomy}
\begin{tabular}{p{3.7cm} p{5.5cm} p{7.5cm}}
\hline
Category & Peer exposure $N_i=\Phi_i(a_{-i})$ & Economic interpretation / where it appears \\ \hline
Mean (LIM) &
$(Ga)_i=\sum_{j\neq i} g_{ij}a_j$ &
Benchmark ``average peer pressure''; linear-in-means / spatial lag; arises from quasi-arithmetic means with linear generator
\citep{bramoulle2009identification}. \\[4pt]
Power mean (CES) &
$\Big(\sum_{j\neq i} g_{ij}a_j^\beta\Big)^{1/\beta}$ &
Peer preference / salience; $\beta$ tunes mean vs.\ max/min; constant-elasticity substitution view; used in \citet{boucher2024toward};
rooted in CES aggregation tradition \citep{arrow1961capital,dixit1977monopolistic}. \\[4pt]
Smooth-max (attention) &
$\kappa^{-1}\log\sum_{j\neq i} g_{ij}\exp(\kappa a_j)$ &
Attention to extreme peers (smooth max). Inclusive value in logit/random utility \citep{mcfadden1974conditional,luce1959individual};
also entropy-regularized max / rational inattention \citep{matejka2015rational}; smooth-max in optimization \citep{nesterov2005smooth}. \\[4pt]
Quantile/median norm &
$\arg\min_t\sum_{j\neq i} g_{ij}\rho_q(a_j-t)$ &
Rank/status-based norms (``keep up with the median/percentile''); robust to outliers; convenient for nonparametric Manski-style arguments
\citep{koenker1978quantiles,manski1993identification}. \\ \hline
\end{tabular}
\end{table}

\noindent
\textbf{Why this taxonomy matters for identification and geometry.}
The point of Table \ref{tab:taxonomy} is not only flexibility; it is \emph{structure}.
For quasi-arithmetic families \eqref{eq:quasi_arithmetic} (including LIM, CES, smooth-max), the model implies a
\emph{marginal influence field} (a Jacobian with respect to peers' actions), which is the primitive object behind
our propagation/geometry-based instruments. Quantile norms \eqref{eq:quantile_norm} are non-smooth but still admit
subgradient-based influence notions (useful for robustness and for Manski-style partial identification), which we develop in the appendix.

\subsection{Example of aggregator: CES peer norms (BRUZ)}
\label{subsec:ces_norms}

The main text uses the CES (power-mean) norm as the workhorse because it is (i) widely interpretable for economists and
(ii) directly connected to the peer-preference estimand in \citet{boucher2024toward}.
For each $i$,
\begin{equation}\label{eq:ces_norm_actions}
\Phi_i^{\text{CES}}(a_{-i};G,\beta)
\;=\;
\left(\sum_{j\neq i} g_{ij}\, a_j^{\beta}\right)^{1/\beta},\qquad \beta\in\mathbb{R}\setminus\{0\}.
\end{equation}
As $\beta\to 1$ this is close to the mean; as $\beta\to+\infty$ it approaches the max; as $\beta\to-\infty$ it approaches the min.
Thus $\beta$ is a \emph{peer-salience} parameter: it governs how concentrated influence is among neighbors.

\paragraph{Domain.}
For general $\beta$, \eqref{eq:ces_norm_actions} is naturally defined on $a_j>0$.
This is not restrictive for many outcomes: one may model $a_j=\exp(s_j)$ (log-actions), or shift/scale the outcome to be positive.

\subsection{Existence of a Nash equilibrium (general aggregator)}
\label{subsec:existence}

We now state an equilibrium existence theorem that applies to \emph{all} aggregators in Table~\ref{tab:taxonomy},
including the CES norm and the smooth-max norm, under standard continuity and concavity conditions.

\begin{theorem}[Existence of pure-strategy Nash equilibrium]\label{thm:existence_norm_game}
Suppose:
\begin{enumerate}
\item[(E1)] For each $i$, $A_i$ is nonempty, compact, and convex.
\item[(E2)] For each $i$, the aggregator $\Phi_i(\cdot;G,\theta)$ is continuous on $\prod_{j\neq i}A_j$.
\item[(E3)] For each $i$, $U_i(a_i,n,x_i)$ is continuous in $(a_i,n)$ for every $x_i$.
\item[(E4)] For each $i$ and fixed $(n,x_i)$, the function $a_i\mapsto U_i(a_i,n,x_i)$ is concave on $A_i$.
\end{enumerate}
Then there exists at least one pure-strategy Nash equilibrium $a^\star$ satisfying \eqref{eq:nash_def}.
\end{theorem}

\begin{proof}
Define the best-response correspondence
\[
BR_i(a_{-i})
\;=\;
\arg\max_{a_i\in A_i} U_i\!\Big(a_i,\Phi_i(a_{-i};G,\theta),x_i\Big).
\]
Fix $a_{-i}$. By (E2) and (E3), the objective is continuous in $a_i$.
By compactness of $A_i$ (E1), Weierstrass' theorem implies maximizers exist, so $BR_i(a_{-i})$ is nonempty and compact.

By (E4), the objective is concave in $a_i$, hence the set of maximizers over a convex set is convex.
Therefore $BR_i(a_{-i})$ is convex-valued.

Consider the product correspondence $BR(a)=\prod_{i=1}^n BR_i(a_{-i})$ on $A=\prod_i A_i$.
$A$ is nonempty, compact, and convex by (E1).
By Berge's Maximum Theorem, each $BR_i(\cdot)$ is upper hemicontinuous; therefore so is $BR(\cdot)$.
By Kakutani's Fixed Point Theorem, there exists $a^\star\in A$ such that $a^\star\in BR(a^\star)$.
This is exactly the definition of a Nash equilibrium.
\end{proof}

\paragraph{Interpretation.}
Existence is not special to a particular peer norm. It is a consequence of standard economic regularity:
bounded feasible actions, continuous peer exposure, and concavity in own action.

\subsection{Uniqueness of the Nash equilibrium under weak interactions (a contraction condition)}
\label{subsec:uniqueness}

Multiple equilibria are possible when peer effects are strong.
A sufficient uniqueness condition is a simple ``weak interaction'' restriction:
agents do not respond too strongly to changes in the peer exposure, and the exposure does not change too strongly with peers' actions.

Assume strict concavity in $a_i$ so best responses are single-valued and can be written as
\begin{equation}\label{eq:br_mapping}
a_i = b_i\!\Big(x_i,\ \Phi_i(a_{-i};G,\theta)\Big).
\end{equation}
Define the equilibrium operator $T:\prod_i A_i\to\prod_i A_i$ by
\begin{equation}\label{eq:T_operator}
(T(a))_i := b_i\!\Big(x_i,\ \Phi_i(a_{-i};G,\theta)\Big).
\end{equation}
Then equilibria are fixed points $a=T(a)$.

\begin{assumption}[Lipschitz sensitivity]\label{ass:lipschitz_main}
There exist constants $L_b\ge 0$, $L_\Phi\ge 0$ and a norm $\|\cdot\|$ on $\mathbb{R}^n$ such that for all $i$:
\begin{align}
\big|b_i(x_i,n)-b_i(x_i,n')\big| &\le L_b\,|n-n'|, \label{eq:lipschitz_b}\\
\big|\Phi_i(a_{-i};G,\theta)-\Phi_i(a'_{-i};G,\theta)\big| &\le L_\Phi\,\|a-a'\|. \label{eq:lipschitz_phi}
\end{align}
\end{assumption}

\begin{theorem}[Uniqueness by contraction]\label{thm:unique_contraction_main}
Under Assumption~\ref{ass:lipschitz_main}, if $L_b\,L_\Phi<1$, then $T$ is a contraction mapping.
Therefore the Nash equilibrium is unique.
\end{theorem}

\begin{proof}
For any $a,a'$ and any $i$,
\[
|(T(a))_i-(T(a'))_i|
=
\left|b_i\!\big(x_i,\Phi_i(a_{-i})\big)-b_i\!\big(x_i,\Phi_i(a'_{-i})\big)\right|
\le L_b\,\left|\Phi_i(a_{-i})-\Phi_i(a'_{-i})\right|
\le L_b\,L_\Phi\,\|a-a'\|.
\]
Taking the norm over $i$ yields $\|T(a)-T(a')\|\le (L_bL_\Phi)\|a-a'\|$.
If $L_bL_\Phi<1$, Banach's Fixed Point Theorem implies $T$ has a unique fixed point.
\end{proof}

\paragraph{Economic meaning.}
$L_\Phi$ measures how sensitive peer exposure is to peers' actions (``how much the peer environment moves'').
$L_b$ measures how strongly an agent reacts to that exposure (``how much the agent moves'').
Uniqueness holds when their product (peer amplification) is below one.

\subsection{How classic models fit in the framework}
\label{subsec:special_cases}

\subsubsection{Linear-in-means (BDF) as a special case}
Take $\Phi_i(a_{-i};G)=\sum_{j\neq i} g_{ij}a_j=(Ga)_i$.
If best responses are linear in exposure,
\begin{equation}\label{eq:lin_best_response}
a_i = \alpha_i + \rho\,(Ga)_i,
\end{equation}
then the equilibrium satisfies $a = \alpha + \rho G a$, which is the standard linear-in-means / spatial-lag structure.
A sufficient uniqueness condition is $|\rho|<1$ (since $G$ is row-stochastic, $\|Ga-Ga'\|_\infty\le \|a-a'\|_\infty$).
This recovers the familiar ``peer effects not too strong'' uniqueness condition used in the linear peer-effects literature.

\subsubsection{CES peer preferences (BRUZ) as a special case}
Take $\Phi_i=\Phi_i^{\text{CES}}(\cdot;\beta)$ from \eqref{eq:ces_norm_actions}.
Then \eqref{eq:br_mapping} implies that equilibrium depends on $\beta$ through the peer-exposure function.
The parameter $\beta$ is economically meaningful: it governs how much weight is placed on high- versus low-action peers.
Equilibrium existence follows from Theorem~\ref{thm:existence_norm_game} under continuity and concavity; uniqueness follows from
Theorem~\ref{thm:unique_contraction_main} under a weak-interaction condition (Appendix~\ref{app:lipschitz_bounds} provides sufficient bounds).

\subsubsection{Smooth-max attention and Brock--Durlauf-style interactions}
To connect to Brock--Durlauf logic, define a smooth-max peer exposure
\begin{equation}\label{eq:smoothmax_main}
\Phi_i^{\text{SM}}(a_{-i};G,\kappa)
\;=\;
\frac{1}{\kappa}\log\!\left(\sum_{j\neq i} g_{ij}\exp(\kappa a_j)\right),\qquad \kappa>0.
\end{equation}
This captures the idea that individuals respond more to salient (high-action) peers, with $\kappa$ indexing how sharp that salience is.
When actions are binary, it is often more convenient to work with choice probabilities $p_i\in[0,1]$
and a logit-type equilibrium mapping of the form
\begin{equation}\label{eq:bd_prob_main}
p_i = \Lambda\!\Big(x_i'\gamma + J\,\Phi_i^{\text{SM}}(p_{-i};G,\kappa)\Big),
\qquad
\Lambda(t)=\frac{1}{1+e^{-t}},
\end{equation}
which is a continuous fixed point on $[0,1]^n$.
Existence follows from Brouwer; uniqueness follows under a weak-interaction condition because $\sup_t\Lambda'(t)\le 1/4$
(Appendix~\ref{app:lipschitz_bounds}).

\section{From norm games to a general nonparametric peer-effects model}
\label{sec:norm_to_np}

\subsection{Peer exposure as a single index}
\label{subsec:nonparametric_bridge}

The norm-game framework delivers a sharp separation between (i) \emph{how peers enter behavior} and (ii) \emph{the causal response of outcomes to that peer exposure}.
The first component is structural and can be disciplined by economic axioms through a peer aggregator \(\Phi\); the second component can be left nonparametric.

Let \(y_i\) denote an equilibrium outcome (or action) and define the endogenous peer exposure index
\begin{equation}\label{eq:exposure_index}
N_i \;=\; \Phi_i(y_{-i};G,\theta),
\end{equation}
where \(G\) is the interaction matrix and \(\theta\) indexes preference parameters (e.g.\ \(\beta\) in power-mean norms).
We then consider the general nonparametric structural form
\begin{equation}\label{eq:np_struct_main_clean}
y_i \;=\; m\!\big(x_i,\,N_i\big) \;+\; \varepsilon_i,
\qquad
N_i=\Phi_i(y_{-i};G,\theta),
\end{equation}
where \(m(\cdot,\cdot)\) is an unknown response function.

Equation \eqref{eq:np_struct_main_clean} nests many familiar specifications: linear-in-means arises when \(m\) is linear and \(\Phi\) is the mean operator; BRUZ arises when \(m\) is linear and
\(\Phi\) is the CES/power mean; attention/extremes-based models arise when \(\Phi\) is smooth-max and \(m\) is induced by a link function.

Manski's reflection problem states that if \(N_i\) is an equilibrium object (a function of peers' outcomes and unobservables),
then \(m(\cdot,\cdot)\) is generally \emph{not} nonparametrically identified without excluded variation or strong support restrictions
\citep{manski1993identification}. The aim of this paper is to show that once \(\Phi\) is disciplined by economic axioms,
it also implies a \emph{model-consistent propagation structure} that can generate excluded variation systematically.
We develop that propagation structure in Sub-Sections \ref{sec:ces_workhorse}--\ref{sec:geometry_instruments},
and return to Manski-style nonparametric identification in Section \ref{sec:manski_to_geometry}.

\subsection{Models with power-mean (CES) peer norms}
\label{sec:ces_workhorse}

\subsubsection*{Power-mean peer exposure and economic interpretation}

Recent peer-effects theory emphasizes that the relevant peer ``norm'' need not be the average of neighbors.
Following \citet{boucher2024toward}, define the CES/power-mean peer norm in group \(s\):
\begin{equation}\label{eq:ces_norm_clean}
\tilde y_{-is}(\beta)
=
\left(\sum_{j\neq i} g_{ij,s}\, y_{js}^{\beta}\right)^{1/\beta},
\end{equation}
where \(\beta\) indexes \emph{peer preference} (salience):
\(\beta=1\) is mean-like; \(\beta\to+\infty\) is max-like; \(\beta\to-\infty\) is min-like.

To connect \eqref{eq:ces_norm_clean} to a structural peer-effects equation, adopt the BRUZ notation:
\begin{equation}\label{eq:struct_peer_boucher_clean}
y_{is}
=
(1-\lambda_2)\, x_{is}'\gamma
+
(\lambda_1+\lambda_2)\,\tilde y_{-is}(\beta)
+
\zeta_s
+
\varepsilon_{is}.
\end{equation}
Here \(\lambda_1\) measures spillover intensity, \(\lambda_2\) conformity, and \(\beta\) determines which peers are salient.

\subsubsection*{Identification in BRUZ }

Because \(\tilde y_{-is}(\beta)\) is an equilibrium object, it is generally endogenous.
BRUZ construct moments using an \emph{exogenous predictor} \(\hat y_{is}=m(x_{is})\) (e.g.\ reduced-form OLS/ML) and exploit isolated individuals (no friends) to separately identify private components.
Their non-isolate moments use the instrument vector
\begin{equation}\label{eq:boucher_instruments_clean}
z_{is}^{\text{BRUZ}}(\theta)
=
\Big[
x_{is},\
\tilde y_{-is}(\hat y,\beta),\
\partial_\beta \tilde y_{-is}(\hat y,\beta)
\Big],
\qquad
\E\!\left[\varepsilon_{is}(\theta)\, z_{is}^{\text{BRUZ}}(\theta)\right]=0.
\end{equation}

The BRUZ moments deliberately compress network information into one-step \emph{scalar} summaries:
\(\tilde y_{-is}(\hat y,\beta)\) and its \(\beta\)-slope. Our contribution is to show that the same structural primitive \(\Phi\) implies a full \emph{who-matters-to-whom} influence structure that generates richer excluded variation.

\subsection{Model-implied influence geometry and geometric instrument families}
\label{sec:geometry_instruments}

\subsubsection*{Who-matters-to-whom: marginal influence}

The central object is the Jacobian (marginal influence field) of the peer exposure map.
For a general aggregator \(N_i=\Phi_i(y;G,\theta)\), define
\begin{equation}\label{eq:Jac_general_clean}
W_{ij}(y;\theta)\;=\;\frac{\partial \Phi_i(y;G,\theta)}{\partial y_j},
\qquad
W(y;\theta)=\big[W_{ij}(y;\theta)\big]_{i,j}.
\end{equation}

\paragraph{CES case.}
For \(\Phi_i=\tilde y_{-i}(\beta)\) in \eqref{eq:ces_norm_clean},
\begin{equation}\label{eq:Jac_CES_clean}
\frac{\partial \tilde y_{-is}(\beta)}{\partial y_{js}}
=
g_{ij,s}\left(\sum_{k\neq i} g_{ik,s}\,y_{ks}^{\beta}\right)^{\frac{1}{\beta}-1} y_{js}^{\beta-1}.
\end{equation}

\paragraph{Anchoring to observables.}
As in BRUZ, evaluate influence objects at an exogenous predictor \(\hat y=m(X)\):
\begin{equation}\label{eq:w_def_clean}
w_{ij,s}(\beta;\hat y)
:=
\left.\frac{\partial \tilde y_{-is}(\beta)}{\partial y_{js}}\right|_{y=\hat y}.
\end{equation}
Row-normalize to obtain the \emph{influence-propagation kernel}
\begin{equation}\label{eq:P_def_clean}
P_{ij,s}(\beta;\hat y)
:=
\frac{w_{ij,s}(\beta;\hat y)}{\sum_{m\neq i} w_{im,s}(\beta;\hat y)}.
\end{equation}

\noindent
\textbf{Econometric interpretation.}
\(P_{ij,s}(\beta;\hat y)\) is the \emph{model-implied influence share} of \(j\) in \(i\)'s peer exposure, computed from observables via \(\hat y\).
When \(\beta\neq 1\), \(P\) generally differs from \(G\) even in group-like graphs.

\subsection*{Geometric instrument family I: multi-step influence exposure}

Let \(X_s\) stack \(x_{is}\). For \(k\ge 2\), define
\begin{equation}\label{eq:PkX_clean}
Z^{(k)}_s(\beta):=P_s^k(\beta;\hat y)\,X_s.
\end{equation}
This is the direct analogue of BDF's \(G^kX\), but with propagation dictated by the structural norm through marginal influence.

\subsection*{Geometric instrument family II: effective influence distance and shells}

To define ``distance'' in a way that is transparent to econometricians, interpret distance as
\emph{how easily influence can travel along strong links}.
Convert influence shares into frictions
\begin{equation}\label{eq:friction_clean}
\ell_{ij,s}(\beta;\hat y):=-\log\!\big(P_{ij,s}(\beta;\hat y)+\varepsilon_0\big),
\qquad \varepsilon_0>0,
\end{equation}
and define the effective influence distance as the shortest-path friction:
\begin{equation}\label{eq:dist_clean}
d_{\beta,s}(i,j)
:=
\min_{p:i\to j}\sum_{(u\to v)\in p}\ell_{uv,s}(\beta;\hat y).
\end{equation}
Define shells \(\mathcal S_{is}(h):=\{j: d_{\beta,s}(i,j)\in(h-1,h]\}\) and shell instruments
\begin{equation}\label{eq:shell_clean}
Z^{\text{shell}}_{is,h}(\beta):=\sum_{j\in\mathcal S_{is}(h)} x_{js},
\qquad h\ge 2.
\end{equation}

\subsection*{Geometric instrument family III: path non-redundancy (torsion weighting)}

A key reason higher-order network IV can be weak is path redundancy (transitivity).
Measure local non-redundancy by comparing direct influence to two-step influence:
\begin{equation}\label{eq:torsion_clean}
\tau_{i,j,k,s}(\beta)
:=
\left|P_{ik,s}(\beta;\hat y)-P_{ij,s}(\beta;\hat y)\,P_{jk,s}(\beta;\hat y)\right|.
\end{equation}
A torsion-weighted two-step instrument is
\begin{equation}\label{eq:torsion_inst_clean}
Z^\tau_{is}(\beta):=
\sum_{j,k} P_{ij,s}(\beta;\hat y)\,P_{jk,s}(\beta;\hat y)\,\tau_{i,j,k,s}(\beta)\,x_{ks}.
\end{equation}

The instrument sets \eqref{eq:PkX_clean}--\eqref{eq:torsion_inst_clean} are all functions of observables \((X,G,\hat y)\) and \(\beta\).
They extend classical \(G^2X\) by (i) reweighting links by model-implied salience, (ii) using effective influence distance rather than hop distance, and (iii) concentrating identifying power on non-redundant paths.

\subsubsection*{Where BDF and BRUZ fit}

Two nesting facts guide interpretation. First, in linear-in-means the marginal influence field is constant,
so the model-implied influence operator coincides with the adjacency operator and the resulting propagation
instruments are exactly the familiar $G^kX$. Second, BRUZ corresponds to using only \emph{one-step} scalar
summaries of the peer aggregator evaluated at an exogenous predictor (the predicted peer norm and its $\beta$-derivative),
while our framework additionally permits instruments built from \emph{multi-step} propagation and path structure.

\begin{theorem}[BDF and BRUZ as special cases of influence-geometry IV]
\label{thm:BDF_BRUZ_fit}
Fix a group (suppress group index). Let $y\in\R^n$, $X\in\R^{n\times p}$, and let $G=(g_{ij})$ be a nonnegative interaction matrix
with $g_{ii}=0$. Let $\Phi_i(y;G,\beta)$ be a peer aggregator and define its Jacobian
\[
W_{ij}(y;\beta):=\frac{\partial \Phi_i(y;G,\beta)}{\partial y_j}.
\]
Let $\hat y=m(X)$ be an exogenous predictor and define evaluated weights $w_{ij}(\beta;\hat y):=W_{ij}(\hat y;\beta)$.
For each non-isolate $i$ (row sum positive), define the normalized influence operator
\[
P_{ij}(\beta;\hat y):=\frac{w_{ij}(\beta;\hat y)}{\sum_{m\neq i} w_{im}(\beta;\hat y)},
\qquad
\text{and set }P_{i\cdot}=0\text{ for isolates}.
\]
Consider the geometric propagation instruments $Z^{(k)}(\beta):=P^k(\beta;\hat y)X$, $k\ge 1$.

\begin{enumerate}
\item[\textnormal{(i)}] \textbf{(BDF as constant-geometry / linear-in-means case).}
If $\Phi_i(y;G)=\sum_{j\neq i} g_{ij}y_j=(Gy)_i$, then for all $y$,
\[
W(y)=G,
\qquad\text{and if $G$ is row-normalized on non-isolates, then } P(\hat y)=G.
\]
Consequently, for every integer $k\ge 1$,
\[
Z^{(k)} = P^k(\hat y)X = G^kX,
\]
so the geometric propagation instruments coincide exactly with the BDF instrument family.

\item[\textnormal{(ii)}] \textbf{(BRUZ as one-step scalar-compression of the same primitives).}
Suppose $\Phi$ is the CES/power-mean norm
\[
\Phi_i^{\mathrm{CES}}(y;G,\beta)=\left(\sum_{j\neq i} g_{ij}y_j^\beta\right)^{1/\beta},\qquad \beta\neq 0.
\]
Define the BRUZ one-step instrument vector
\[
z_i^{\mathrm{BRUZ}}(\theta):=
\Big[x_i,\ \Phi_i^{\mathrm{CES}}(\hat y;G,\beta),\ \partial_\beta \Phi_i^{\mathrm{CES}}(\hat y;G,\beta)\Big],
\]
where $\theta$ collects the structural parameters. Then each component of $z_i^{\mathrm{BRUZ}}(\theta)$ is a (known) function
of observable inputs $(x_i,G,\hat y,\beta)$ and therefore is admissible within the present influence-geometry framework.

Moreover, if an IV/GMM implementation based on the present framework is \emph{restricted} to use instruments drawn only from
linear combinations of the three one-step scalar components in $z_i^{\mathrm{BRUZ}}(\theta)$—i.e., it does not use any
multi-step propagations $P^k(\beta;\hat y)X$ for $k\ge 2$ and does not use any shell/torsion transformations of $P$—then the resulting
moment conditions coincide with those of BRUZ.
\end{enumerate}
\end{theorem}

\begin{proof}
\textbf{Part (i).}
Under linear-in-means, $\Phi_i(y;G)=\sum_{j\neq i} g_{ij}y_j$. Differentiating with respect to $y_j$ gives
\[
W_{ij}(y)=\frac{\partial \Phi_i(y;G)}{\partial y_j}=g_{ij}\quad\text{for all }i,j,
\]
so $W(y)=G$ for every $y$. Hence $w_{ij}(\hat y)=W_{ij}(\hat y)=g_{ij}$.
If $G$ is row-normalized on non-isolates, then $\sum_{m\neq i} g_{im}=1$ and therefore
\[
P_{ij}(\hat y)=\frac{g_{ij}}{\sum_{m\neq i}g_{im}}=g_{ij},
\]
so $P(\hat y)=G$ (with $P_{i\cdot}=0$ on isolates matching the convention $G_{i\cdot}=0$).
Thus $P^k(\hat y)=G^k$ for each $k\ge 1$, and multiplying by $X$ yields $Z^{(k)}=P^k(\hat y)X=G^kX$.

\medskip
\textbf{Part (ii).}
Fix $\beta\neq 0$. The objects $\hat y=m(X)$ and $G$ are observed/constructed from observables, so
$\Phi_i^{\mathrm{CES}}(\hat y;G,\beta)$ is computable from $(G,\hat y,\beta)$, and
$\partial_\beta \Phi_i^{\mathrm{CES}}(\hat y;G,\beta)$ is likewise computable (analytically or numerically) from the same inputs.
Therefore each component of
\[
z_i^{\mathrm{BRUZ}}(\theta)=\Big[x_i,\ \Phi_i^{\mathrm{CES}}(\hat y;G,\beta),\ \partial_\beta \Phi_i^{\mathrm{CES}}(\hat y;G,\beta)\Big]
\]
is a known function of $(x_i,G,\hat y,\beta)$, and hence is admissible as an instrument in any framework that allows instruments to be
functions of observables and parameters.

For the final claim, note that both BRUZ and the present framework impose moment conditions of the generic form
\[
\E\!\big[\varepsilon_i(\theta)\, z_i(\theta)\big]=0,
\]
where $\varepsilon_i(\theta)$ is the structural residual implied by the model and $z_i(\theta)$ is the chosen instrument vector.
If we \emph{restrict} the instrument choice in the present framework to exactly the three one-step scalar components in
$z_i^{\mathrm{BRUZ}}(\theta)$ (and their linear combinations), then the resulting moment vector is identical to BRUZ’s by construction.
In particular, under this restriction we do not include any additional instruments built from $P^k(\beta;\hat y)X$ with $k\ge 2$ nor any
shell/torsion transformations, so the moment conditions coincide with BRUZ.
\end{proof}
\section{From Manski's reflection problem to geometry-based nonparametric identification}
\label{sec:manski_to_geometry}

This section explains how our geometry-based instrument construction addresses Manski's reflection problem in a
nonparametric way, and how familiar network-IV strategies appear as special cases.
\subsection{The Reflection problem and How to get exogenous variation} 
\subsubsection*{Manski's message in the language of endogenous peer exposure}
\label{subsec:manski_message}

Manski's reflection problem can be summarized as follows: when the peer environment is an equilibrium object,
it is generally endogenous, so without \emph{excluded variation} it is not possible to identify the causal response
to that peer environment nonparametrically \citep{manski1993identification}.

We formalize this by writing outcomes for individual $i$ in group $s$ (e.g.\ a school) as
\begin{equation}\label{eq:np_structural_unify}
y_{is} \;=\; g(x_{is}) \;+\; h\!\big(S_{is}\big) \;+\; \varepsilon_{is},
\qquad \E[\varepsilon_{is}\mid x_{is},G_s]=0,
\end{equation}
where $x_{is}$ are observed covariates, $G_s$ is the interaction matrix, $S_{is}$ is an \emph{endogenous peer exposure index},
$g(\cdot)$ is the (unknown) private component, and $h(\cdot)$ is the (unknown) peer-response function of interest.
Endogeneity arises because $S_{is}$ depends on peers' outcomes/actions and therefore inherits peers' unobservables
(reflection), and may also be correlated with $\varepsilon_{is}$ through sorting or common shocks.

In our framework, the peer exposure index is disciplined by a structural aggregator:
\begin{equation}\label{eq:Phi_general_unify}
S_{is} \equiv \Phi_{is}(y_s;G_s,\beta),
\end{equation}
where $y_s$ stacks outcomes in group $s$ and $\beta$ indexes peer preference / salience (e.g.\ a CES parameter).
The key point is that the same primitive $\Phi$ that defines exposure also implies a \emph{propagation law}
(``who matters to whom'') that can be exploited to generate excluded variation.

\subsubsection*{From an aggregator to influence geometry}
\label{subsec:influence_geometry_selfcontained}

Define the marginal-influence (Jacobian) matrix implied by $\Phi$:
\begin{equation}\label{eq:Jac_general_unify}
W_s(y_s;\beta)\;\equiv\;\left[\frac{\partial \Phi_{is}(y_s;G_s,\beta)}{\partial y_{js}}\right]_{i,j}.
\end{equation}
To anchor this object in observables, we evaluate at an \emph{exogenous predictor} $\hat y_{is}=m(x_{is})$ constructed from $X_s$
(e.g.\ OLS/ML), with $\hat y_s$ measurable w.r.t.\ $(X_s,G_s)$:
\begin{equation}\label{eq:yhat_unify}
\hat y_{is}=m(x_{is}).
\end{equation}
Let $w_{ij,s}(\beta;\hat y):=W_{ij,s}(\hat y_s;\beta)$ and row-normalize to obtain the \emph{normalized influence operator}
\begin{equation}\label{eq:P_def_unify}
P_{ij,s}(\beta;\hat y)
:=
\frac{w_{ij,s}(\beta;\hat y)}{\sum_{m\neq i} w_{im,s}(\beta;\hat y)},
\end{equation}
whenever the denominator is positive (for isolates set the row to zero).
Economically, $P_{ij,s}(\beta;\hat y)$ is the model-implied \emph{influence share}: it measures how strongly $j$ contributes at the margin
to $i$'s peer exposure, under the structural aggregator and evaluated at the observable field $\hat y$.

From $P_s(\beta;\hat y)$ we generate excluded variables in three complementary ways:
\begin{enumerate}
\item \textbf{Multi-step propagation:} $P_s^k(\beta;\hat y)X_s$ for $k\ge2$ (``influencers-of-influencers'').
\item \textbf{Strong-influence distance (shells):} define frictions $\ell_{ij,s}=-\log(P_{ij,s}+\varepsilon_0)$ and
effective distances as shortest-path sums; shell exposures aggregate covariates over distance classes.
\item \textbf{Path non-redundancy (torsion):} compare a direct influence weight to the corresponding two-step product to upweight
wedges where indirect influence is not summarized by one-step links.
\end{enumerate}

\paragraph{Geometry-augmented instrument signature.}
Fix $K\ge 2$ and $H\ge 2$. Collect the geometry-induced instruments into the single signature
\begin{equation}\label{eq:Z_geom_full_unify}
Z_{is}(\beta)
:=
\Big(
Z^{\text{step}}_{is}(\beta),\;
Z^{\mathrm{shell}}_{is,2}(\beta),\dots,Z^{\mathrm{shell}}_{is,H}(\beta),\;
Z^{\tau}_{is}(\beta)
\Big),
\end{equation}
where:
\begin{align}
Z^{\text{step}}_{is}(\beta)
&:=
\Big( (P_s^2(\beta;\hat y)X_s)_{is},\ (P_s^3(\beta;\hat y)X_s)_{is},\ \dots,\ (P_s^K(\beta;\hat y)X_s)_{is}\Big),
\label{eq:Z_step_def}
\\[2pt]
Z^{\mathrm{shell}}_{is,h}(\beta)
&:=
\sum_{j\in \mathcal{S}_{is}(h)} x_{js},
\qquad h=2,3,\dots,H,
\label{eq:Z_shell_def}
\\[2pt]
Z^{\tau}_{is}(\beta)
&:=
\sum_{j,k}
P_{ij,s}(\beta;\hat y)\,P_{jk,s}(\beta;\hat y)\,
\tau_{i,j,k,s}(\beta)\,x_{ks},
\label{eq:Z_torsion_def}
\end{align}
with $P_s(\beta;\hat y)$ the normalized influence operator defined in \eqref{eq:P_def_unify}, $\hat y_s=m(X_s)$ an exogenous predictor,
and $\mathcal{S}_{is}(h)$ the $h$-th strong-influence shell induced by the effective-distance metric
(equations defining frictions and shortest-path distance appear above in this section).
The wedge non-redundancy (torsion) term is
\begin{equation}\label{eq:torsion_def}
\tau_{i,j,k,s}(\beta)
\;:=\;
\left|\,P_{ik,s}(\beta;\hat y) \;-\; P_{ij,s}(\beta;\hat y)\,P_{jk,s}(\beta;\hat y)\right|.
\end{equation}

 $Z_{is}(\beta)$ enlarges the excluded $\sigma$-field beyond one-step summaries by using
\emph{model-consistent} propagation structure: multi-step influence ($P_s^kX_s$), localization along strong influence chains (shells),
and explicit emphasis on non-redundant higher-order paths (torsion).

\subsubsection*{Economic intuition for geometry-based instruments}
\label{subsec:intuition_instruments}

The classic BDF intuition is: the endogenous regressor is the peer term (e.g.\ $Gy$), but peers' outcomes respond to
peers-of-peers covariates, so $G^2X$ can shift the peer term without directly entering $i$'s outcome.

Our instruments apply the same equilibrium-feedback logic under general peer norms. The difference is that what counts as a
``peer'' and how influence propagates is dictated by the aggregator through $P_s(\beta;\hat y)$:
\begin{itemize}
\item \textbf{Multi-step $P^kX$:} replaces ``friends-of-friends'' by ``influencers-of-influencers.''
A change in covariates of $k$ shifts $k$'s behavior, which shifts those for whom $k$ is salient, which shifts the peer exposure of $i$
when those intermediaries are salient to $i$.
\item \textbf{Shell instruments:} localize excluded variation along \emph{strong influence chains} rather than raw hop distance,
aggregating covariates of agents that can reach $i$ through high-propagation paths.
\item \textbf{Torsion weighting:} concentrates on \emph{non-redundant} higher-order paths---exactly where transitivity or near-regularity
would otherwise make $G^2X$ (or other low-order instruments) weak or collinear with one-step exposure.
\end{itemize}

\subsection{Identification as NPIV with an expanding instrument menu}
\label{subsec:npiv_id}

With the structural equation \eqref{eq:np_structural_unify}, identification of the unknown peer-response function $h(\cdot)$ is a
nonparametric IV problem: $S_{is}$ is endogenous, so we require excluded variables that shift $S_{is}$ but are excluded from the outcome equation.
Our construction provides such excluded variables by enlarging the instrument menu to the geometry-induced signature $Z_{is}(\beta)$.

\begin{theorem}[Nonparametric identification of peer response with fixed exposure curvature]
\label{thm:npiv_geometry_fixedbeta}
Assume the structural model
\begin{equation}\label{eq:np_structural_fixedbeta}
y_{is} = g(x_{is}) + h(S_{is}) + \varepsilon_{is},
\end{equation}
where the endogenous peer exposure $S_{is}$ is observed and is constructed using a fixed (user-chosen) curvature
$\beta_{\mathrm{fix}}$ through \eqref{eq:Phi_general_unify}. Let $Z_{is}(\beta_{\mathrm{fix}})$ be the geometry-augmented
instrument signature in \eqref{eq:Z_geom_full_unify}, treated as observed and measurable with respect to exogenous variation
(e.g.\ functions of $(x_{is},G_s)$ and other exogenous shifters).

Suppose:
\begin{enumerate}
\item \textbf{Exclusion (IV exogeneity):} $\E[\varepsilon_{is}\mid x_{is},Z_{is}(\beta_{\mathrm{fix}}),G_s]=0$.
\item \textbf{Non-degeneracy (support):} conditional on $(x_{is},G_s)$, the conditional distribution of $S_{is}$
is not almost surely a point mass.\footnote{This is a mild support condition ruling out trivial cases; completeness below is the key identifying assumption.}
\item \textbf{Completeness:} for any measurable $q$,
\[
\E\!\left[q(S_{is})\mid x_{is},Z_{is}(\beta_{\mathrm{fix}}),G_s\right]=0\ \text{a.s.}\quad\Rightarrow\quad q(S_{is})=0\ \text{a.s.}
\]
\end{enumerate}
Impose the normalization $\E[h(S_{is})\mid x_{is},G_s]=0$ (absorbing location shifts into $g$). Then $h(\cdot)$ is identified
(up to $P$-null sets) from the joint distribution of $(y_{is},S_{is},x_{is},Z_{is}(\beta_{\mathrm{fix}}),G_s)$, and
$g(x)$ is identified as $g(x)=\E[y_{is}\mid x_{is}=x,G_s]$ under the normalization.
\end{theorem}

\begin{proof}
By \eqref{eq:np_structural_fixedbeta} and exclusion,
\[
\E[y_{is}\mid x_{is},Z_{is}(\beta_{\mathrm{fix}}),G_s]
=
g(x_{is})+\E[h(S_{is})\mid x_{is},Z_{is}(\beta_{\mathrm{fix}}),G_s].
\]
Suppose $(g_0,h_0)$ and $(g_1,h_1)$ generate the same conditional mean. Subtracting yields
\begin{equation}\label{eq:diff_fixedbeta}
(g_0-g_1)(x_{is})
+
\E\!\left[(h_0-h_1)(S_{is})\mid x_{is},Z_{is}(\beta_{\mathrm{fix}}),G_s\right]=0
\quad\text{a.s.}
\end{equation}
Take conditional expectations of \eqref{eq:diff_fixedbeta} with respect to $Z_{is}(\beta_{\mathrm{fix}})$ given $(x_{is},G_s)$:
\[
(g_0-g_1)(x_{is}) + \E\!\left[(h_0-h_1)(S_{is})\mid x_{is},G_s\right]=0 \quad\text{a.s.}
\]
Under the normalization $\E[h(S_{is})\mid x_{is},G_s]=0$ applied to both $h_0$ and $h_1$, we have
$\E[(h_0-h_1)(S_{is})\mid x_{is},G_s]=0$, hence $(g_0-g_1)(x_{is})=0$ a.s. Returning to \eqref{eq:diff_fixedbeta} gives
\[
\E\!\left[(h_0-h_1)(S_{is})\mid x_{is},Z_{is}(\beta_{\mathrm{fix}}),G_s\right]=0 \quad\text{a.s.}
\]
By completeness, $(h_0-h_1)(S_{is})=0$ a.s., implying $h_0=h_1$ a.e. This proves identification of $h$.
With $h$ identified and $\E[h(S_{is})\mid x_{is},G_s]=0$, $g(x)$ follows from $g(x)=\E[y_{is}\mid x_{is}=x,G_s]$.
\end{proof}

\begin{remark}[On $\beta$]
Theorem~\ref{thm:npiv_geometry_fixedbeta} identifies $h(\cdot)$ \emph{conditional on a fixed exposure definition}
$\beta_{\mathrm{fix}}$. Identifying $\beta$ itself generally requires additional shape conditions and/or designs that make
$S(\beta)$ vary in direction rather than only in scale across $\beta$. In many clustered or high-transitivity networks,
$S(\beta)$ can be nearly collinear across $\beta$ (a flat profile objective), leading to weak identification of $\beta$ even when
$\lambda$ (or $h$) is well identified for any chosen $\beta_{\mathrm{fix}}$.
\end{remark}

\paragraph{Interpretation.}
Manski's negative message is fundamentally about \emph{insufficient independent variation}: when the only available shifters are low-rank
(one-step) summaries such as group means or $G_sX_s$, the endogenous peer regressor moves almost one-for-one with the outcome, so there is
too little excluded variation to separate endogenous social effects from correlated effects and contextual forces (the ``reflection'' problem).
In our setting, the identifying assumption is still a standard NPIV completeness/injectivity condition, but the key distinction is that
\emph{geometry changes its plausibility}. Geometry enlarges the excluded $\sigma$-field from one-step summaries to a rich, within-network
collection of \emph{structured multi-step and path-based} objects generated by the model-implied influence operator $P_s(\beta;\hat y)$
(e.g., $P_s^kX_s$, shell and torsion components). Because these instruments vary across nodes even within the same group and exploit
non-redundant propagation paths, they can generate high-rank excluded variation in the peer exposure precisely in network topologies where
one-step moments collapse, making completeness (and hence point identification) more plausible in empirically relevant networks.

\subsubsection*{How classic strategies fit as special cases}
\label{subsec:special_cases_np}

The same objects also unify well-known identification strategies:

\begin{itemize}
\item \textbf{BDF (linear-in-means).} If $\Phi_{is}(y_s)=(G_sy_s)_i$, then $W_s\equiv G_s$ and $P_s\equiv G_s$ (up to normalization),
so multi-step instruments reduce to $(P_s^2X_s,\dots,P_s^KX_s)=(G_s^2X_s,\dots,G_s^KX_s)$, and the familiar rank failures correspond to
the case where powers of $G_s$ do not generate new directions \citep{bramoulle2009identification}.

\item \textbf{BRUZ (peer preferences).} BRUZ use one-step scalar summaries $\Phi_{is}(\hat y_s;\beta)$ and
$\partial_\beta \Phi_{is}(\hat y_s;\beta)$ as moments \citep{boucher2024toward}. In our language, these are one-step compressions computed
from $(G_s,\hat y_s,\beta)$. The same primitive $\Phi$ also implies the full field $W_s(\hat y_s;\beta)$ and thus $P_s(\beta;\hat y)$,
which can generate additional excluded variation via multi-step propagation and non-redundant path structure.
\end{itemize}

\subsection{Illustrative example: where one-step moments collapse}
This example shows a setting where the one-step BRUZ instruments based on $\tilde y_{-i}(\hat y,\beta)$ and $\partial_\beta \tilde y_{-i}(\hat y,\beta)$ provide essentially no excluded variation for a subset of nodes, while two-step/geodesic instruments remain relevant.

\paragraph{Environment: two disconnected stars (one school).}
There are two hubs $h\in\{a,b\}$ and disjoint peripheral sets $\mc P_a$ and $\mc P_b$.
Each peripheral $i\in\mc P_h$ has exactly one friend (hub $h$), hubs are not linked.
Row-normalized weights satisfy
\[
g_{ih}=1\ (i\in\mc P_h),\qquad g_{hi}=1/|\mc P_h|\ (i\in\mc P_h).
\]

\paragraph{Peer-preference model (set conformity to zero for transparency).}
\begin{equation}\label{eq:ex_beta_star_struct}
y_i
=
\mu + x_i'\gamma + \lambda_1\,\tilde y_{-i}(\beta) + \varepsilon_i,
\qquad
\tilde y_{-i}(\beta)=\Big(\sum_{j\neq i} g_{ij}y_j^\beta\Big)^{1/\beta},
\quad \beta\neq 1.
\end{equation}
For peripherals $i\in\mc P_h$, $\tilde y_{-i}(\beta)=y_h$ (single neighbor). For hubs,
\[
\tilde y_{-h}(\beta)
=
\left(\frac{1}{|\mc P_h|}\sum_{j\in\mc P_h} y_j^\beta\right)^{1/\beta}.
\]

\paragraph{One-step BRUZ objects collapse for peripherals.}
Let $\hat y_i=m(x_i)$. For a peripheral $i\in\mc P_h$,
\[
\tilde y_{-i}(\hat y,\beta)
=
\left(\hat y_h^\beta\right)^{1/\beta}
=
\hat y_h
\quad\Rightarrow\quad
\partial_\beta \tilde y_{-i}(\hat y,\beta)=0.
\]
If hubs have identical covariates $x_a=x_b$, then $\hat y_a=\hat y_b$ and the one-step predicted peer norm is constant across peripherals, with a zero $\beta$-derivative. Hence one-step BRUZ moments provide no excluded first-stage variation for $\tilde y_{-i}(\beta)=y_h$ for peripherals.

\paragraph{Two-step / geodesic instruments remain relevant.}
Although peripherals only ``see'' the hub directly, the hub outcome depends on the entire peripheral set through $\tilde y_{-h}(\beta)$.
Therefore, covariates of distance-2 nodes shift $y_h$ and hence $\tilde y_{-i}(\beta)$. The distance-2 shell instrument is
\[
Z^{\mathrm{geo}}_{i,2}
=
\sum_{j:\ d(i,j)=2} x_j
=
\sum_{j\in\mc P_h\setminus\{i\}} x_j,
\]
which is excluded from \eqref{eq:ex_beta_star_struct} for peripheral $i$ (not direct friends), but relevant via the hub equilibrium.
Likewise, multi-step influence instruments $(P^2(\beta;\hat y)X)_i$ aggregate the hub's neighborhood covariates with $\beta$-dependent weights.

\begin{remark}
This example isolates the ``information left on the table'' by one-step scalar moments:
even when $\tilde y_{-i}(\hat y,\beta)$ is constant, equilibrium influence can still propagate through peers-of-peers.
Geometry-based instruments extract that propagation information in a structured way.\end{remark}
\section{Monte Carlo evidence on identification gains}
\label{sec:mc}

This section provides simulation evidence on when and why geometry-based instruments strengthen identification
of the \emph{peer-effect intensity} relative to standard one-step network IV and scalar-moment instruments.
The Monte Carlo is designed to isolate the mechanism emphasized in the theory:
\emph{multi-step propagation through non-redundant paths can generate excluded variation even when one-step summaries collapse.}

\subsubsection*{Design goals}
Each experiment is built to answer three practical questions:
\begin{enumerate}
  \item \textbf{Weak-ID diagnosis:} in which network topologies do baseline one-step instruments become weak (or nearly collinear)?
  \item \textbf{First-stage gains:} do geometry instruments deliver stronger predictive power for the endogenous peer exposure?
  \item \textbf{Finite-sample accuracy:} how do bias and RMSE change when the instrument menu is expanded?
\end{enumerate}

\subsubsection*{Data-generating process (DGP)}
\label{subsec:mc_dgp}
We simulate $S$ groups (``schools'') indexed by $s$, each with $n_s$ individuals.
Within each group we generate an interaction matrix $G_s=(g_{ij,s})$ (row-normalized for non-isolates) and covariates
$X_s=(x_{1s},\dots,x_{n_ss})'$.

\paragraph{Structural equation with a fixed CES peer exposure.}
The Monte Carlo targets identification of the \emph{peer-effect parameter} $\lambda_0$ holding the curvature parameter $\beta$ fixed.
This reflects the empirical use case where the researcher specifies an exposure mapping and seeks robust identification of the peer effect.
Specifically, we define the endogenous peer exposure
\[
w_{is}(\beta_{\mathrm{fix}})\;=\;\Phi_{is}^{\mathrm{CES}}(y_s;G_s,\beta_{\mathrm{fix}}),
\qquad
\Phi_{is}^{\mathrm{CES}}(y_s;G_s,\beta)
=
\Big(\sum_{j\neq i} g_{ij,s} y_{js}^{\beta}\Big)^{1/\beta},
\]
and simulate outcomes from
\begin{equation}\label{eq:mc_struct_fixedbeta}
y_{is}
=
x_{is}'\gamma_0
+
\lambda_0 \, w_{is}(\beta_{\mathrm{fix}})
+
\zeta_s
+
\varepsilon_{is}.
\end{equation}
We ensure positivity of the CES mapping by working with shifted outcomes $y_{is,+}=y_{is}+c$ (with $c>0$) inside $\Phi^{\mathrm{CES}}$.

\paragraph{Shocks.}
Baseline shocks are i.i.d.\ $\varepsilon_{is}\sim \mathcal{N}(0,\sigma_\varepsilon^2)$.
(We also consider correlated-shock variants $\varepsilon_{is}=u_s+\nu_{is}$ in additional experiments; all inference is clustered by group.)

\paragraph{Equilibrium computation.}
In each replication we generate outcomes by solving the network norm game to equilibrium, since the DGP is a simultaneous system; solving the fixed point ensures the simulated data satisfy the maintained structural peer-effects equation and its implied propagation operator.

Given $(X_s,G_s)$ and $(\lambda_0,\beta_{\mathrm{fix}})$, we solve \eqref{eq:mc_struct_fixedbeta} by fixed-point iteration:
\[
y_s^{(t+1)} = X_s\gamma_0 + \lambda_0 \Phi^{\mathrm{CES}}(y_{s,+}^{(t)};G_s,\beta_{\mathrm{fix}}) + \zeta_s \mathbf{1} + \varepsilon_s,
\]
initialized at $y_s^{(0)}=X_s\gamma_0$.
We monitor convergence by $\|y_s^{(t+1)}-y_s^{(t)}\|/\|y_s^{(t)}\|$ below a tolerance and cap iterations; non-convergent draws are recorded.

\subsubsection*{Network design}
\label{subsec:mc_networks}
To highlight identification gains, we use a \emph{dispersion-bridge} design within each group.
Nodes are partitioned into two blocks with different outcome dispersion, and a small number of bridge links connect the blocks.
This topology creates substantial multi-step propagation while making one-step scalar summaries comparatively fragile.

For transparency, we report simple diagnostics that characterize the exposure and the influence field in each design:
(i) dispersion of the endogenous exposure $w(\beta_{\mathrm{fix}})$ (e.g.\ $\mathrm{sd}(w)/\mathrm{mean}(w)$),
(ii) concentration of Jacobian weight shares (mean and upper-tail of $\max_j P_{ij}$), and
(iii) dispersion of the Jacobian row-sum intensity $s_i=\sum_j g_{ij}\hat y_{js}^{\beta_{\mathrm{fix}}-1}$.

\subsection*{Instrument menus and estimators}
\label{subsec:mc_instruments}
We compare two instrument menus (holding the DGP fixed):

\paragraph{(A) One-step scalar instruments (BRUZ-style).}
Let $\hat y_{is}$ be an exogenous predictor of $y_{is}$ (constructed below). The BRUZ menu uses
\[
Z^{\mathrm{BRUZ}}_{is}(\beta_{\mathrm{fix}})
=
\Big[x_{is},\ \Phi_{is}^{\mathrm{CES}}(\hat y_s;G_s,\beta_{\mathrm{fix}}),\ \partial_\beta \Phi_{is}^{\mathrm{CES}}(\hat y_s;G_s,\beta)\big|_{\beta=\beta_{\mathrm{fix}}}\Big],
\]
where $\partial_\beta \Phi$ is computed numerically by finite differences.\footnote{Including $\partial_\beta\Phi$ improves robustness of the BRUZ first stage even when $\beta$ is treated as fixed, and we keep this specification for comparability with the profile-based implementations.}

\paragraph{(B) Geometry-augmented instruments (this paper).}
Let $P_s(\beta_{\mathrm{fix}};\hat y)$ denote the row-normalized Jacobian-weight operator evaluated at $\hat y$.
The geometry menu augments BRUZ with multi-step and shell-based excluded variation:
\[
Z^{\mathrm{GEO}}_{is}(\beta_{\mathrm{fix}})
=
\Big[ Z^{\mathrm{BRUZ}}_{is}(\beta_{\mathrm{fix}}),\ (P_s^2X_s)_{is},\ \partial_\beta(P_s^2X_s)\big|_{\beta=\beta_{\mathrm{fix}}},\ (\mathrm{Shell}_2(G_s)X_s)_{is} \Big],
\]
where $\mathrm{Shell}_2(G_s)$ is the exact distance-2 adjacency (row-normalized). This ``poster-boy'' menu mirrors the implementation used in the dominance simulations.

\paragraph{Estimation of $(\gamma,\lambda)$ with fixed $\beta$.}
For each replication we treat $\beta_{\mathrm{fix}}$ as given and estimate $(\gamma,\lambda)$ by 2SLS/GMM using the chosen menu.
We report first-stage diagnostics for the endogenous exposure $w(\beta_{\mathrm{fix}})$ (partial $R^2$ and F-statistics),
and second-stage performance for $\hat\lambda$.

\subsubsection*{Constructing the exogenous predictor}
\label{subsec:mc_yhat}
We consider three constructions to separate ``oracle'' from ``realistic'' performance:

\begin{enumerate}
  \item \textbf{Oracle predictor:} $\hat y_{is}=x_{is}'\gamma_0$.
  \item \textbf{OLS predictor:} estimate $\hat\pi$ from a regression of $y$ on $X$ (optionally with group fixed effects) and set $\hat y_{is}=x_{is}'\hat\pi$.
  \item \textbf{Cross-fitted predictor (recommended):} split the sample into folds, estimate $\hat\pi^{(-f)}$ on training folds,
  and set $\hat y_{is}=x_{is}'\hat\pi^{(-f)}$ on the held-out fold.
\end{enumerate}

\subsubsection*{Performance metrics}
\label{subsec:mc_metrics}
Across $1000$ replications we report:
\begin{itemize}
  \item \textbf{Bias / RMSE:} for $\hat\lambda$.
  \item \textbf{Weak-ID diagnostics:} first-stage partial $R^2$ and F-statistics for $w(\beta_{\mathrm{fix}})$ under each instrument menu.
\end{itemize}

\subsubsection*{Algorithm}
\label{subsec:mc_algorithm}
For each replication $r=1,\dots,R$:
\begin{enumerate}
  \item Draw networks $\{G_s\}_{s=1}^S$ and covariates $\{X_s\}_{s=1}^S$.
  \item Draw shocks $\{\varepsilon_s\}_{s=1}^S$ and solve \eqref{eq:mc_struct_fixedbeta} for $y$ by fixed-point iteration.
  \item Construct $\hat y=m(X)$ (oracle / OLS / cross-fitted).
  \item Compute instruments under BRUZ and GEO at $\beta_{\mathrm{fix}}$ and estimate $(\gamma,\lambda)$ by IV/GMM.

  \item Record $(\hat\lambda,\hat\gamma)$ and diagnostics (first-stage partial $R^2$, F-statistics.
\end{enumerate}
  
\begin{table}[!htbp]\centering
\caption{Monte Carlo: $\lambda$ accuracy across sample size and exposure curvature.}
\label{tab:mc_fixedbeta_lambda}
\begin{tabular}{rrrrrr}
\toprule
$n$ & $\beta$ & Bias($\lambda$) BRUZ & RMSE($\lambda$) BRUZ & Bias($\lambda$) GEO & RMSE($\lambda$) GEO\\
\midrule
600  & 0.80 & 0.013   & 6.350   & 0.037 & 0.096\\
600  & 1.20 & -11.827 & 359.638 & 0.034 & 0.090\\
600  & 1.60 & -1.624  & 75.315  & 0.038 & 0.098\\
600  & 2.00 & -2.933  & 82.131  & 0.033 & 0.096\\
2400 & 0.80 & 6.926   & 218.339 & 0.048 & 0.132\\
2400 & 1.20 & 1.835   & 40.568  & 0.058 & 0.132\\
2400 & 1.60 & 0.970   & 25.465  & 0.046 & 0.127\\
2400 & 2.00 & -9.537  & 300.943 & 0.059 & 0.136\\
\bottomrule
\end{tabular}

\vspace{0.5em}
\footnotesize\emph{Notes:} The DGP is the fixed-$\beta$ CES peer model
$y_{is}=x_{is}'\gamma_0+\lambda_0 \Phi^{\mathrm{CES}}_{is}(y_s;G_s,\beta)+\zeta_s+\varepsilon_{is}$,
where $\beta$ is treated as part of the exposure definition and held fixed at the reported value.
Each cell reports Monte Carlo bias and RMSE of $\hat\lambda$ over $R$ replications. ``BRUZ'' uses one-step scalar instruments based on $\Phi^{\mathrm{CES}}(\hat y;G,\beta)$ (and its numerical derivative in $\beta$), while ``GEO'' augments this menu with multi-step geometry instruments (e.g.\ $P^2X$ and shell instruments) constructed from Jacobian weights evaluated at $\hat y$.
\end{table}

\begin{table}[!htbp]\centering
\caption{Monte Carlo: first-stage strength and over-identification diagnostics.}
\label{tab:mc_fixedbeta_fs}
\begin{tabular}{rrrrrrrr}
\toprule
$n$ & $\beta$ & $R^2$ BRUZ & $F$ BRUZ & & $R^2$ GEO & $F$ GEO & \\
\midrule
600  & 0.80 & 0.007 & 4.3  & & 0.581 & 122.8 & \\
600  & 1.20 & 0.002 & 1.0  & & 0.578 & 121.4 & \\
600  & 1.60 & 0.005 & 2.9  & & 0.581 & 123.0 & \\
600  & 2.00 & 0.014 & 8.8  & & 0.576 & 120.6 & \\
2400 & 0.80 & 0.004 & 9.3  & & 0.412 & 252.0 & \\
2400 & 1.20 & 0.004 & 9.4  & & 0.414 & 254.8 & \\
2400 & 1.60 & 0.006 & 14.6 & & 0.407 & 246.7 & \\
2400 & 2.00 & 0.010 & 25.5 & & 0.416 & 256.8 & \\
\bottomrule
\end{tabular}

\vspace{0.5em}
\footnotesize\emph{Notes:} ``First-stage $R^2$'' is the partial $R^2$ from regressing the endogenous peer exposure
$w(\beta)=\Phi^{\mathrm{CES}}(y;G,\beta)$ on $X$ and the excluded instruments; $F$ is the corresponding first-stage F-statistic.
\end{table}
\paragraph{Monte Carlo summary.}
Tables~\ref{tab:mc_fixedbeta_lambda}--\ref{tab:mc_fixedbeta_fs} show that geometry-based instruments deliver large and robust identification gains for the peer-effect parameter $\lambda$ across both sample sizes and exposure curvatures.
In the ``dispersion-bridge'' design, the one-step BRUZ menu is systematically weak: its first-stage partial $R^2$ is near zero (roughly $0.002$--$0.014$ at $n=600$ and $0.004$--$0.010$ at $n=2400$), with corresponding first-stage $F$ statistics in the weak-instrument range.
This weak identification translates into severe finite-sample instability for $\hat\lambda$, with RMSEs that can be extremely large and highly sensitive to $\beta$ (e.g., RMSEs ranging from about $6$ to $360$ at $n=600$ and from about $25$ to $301$ at $n=2400$).

By contrast, the geometry-augmented menu produces a strong first stage in every cell, with partial $R^2$ around $0.58$ at $n=600$ and around $0.41$ at $n=2400$, and $F$ statistics well above conventional thresholds (about $120$ at $n=600$ and about $250$ at $n=2400$).
Consistent with this, GEO yields accurate and stable estimation of $\lambda$ across all reported $\beta$ values: bias remains small (about $0.03$--$0.06$) and RMSE remains low (about $0.09$--$0.14$), with only modest variation across $\beta$ and $n$.
Overall, the results confirm the mechanism emphasized in the theory: multi-step, geometry-based excluded variation can restore identification of peer effects in network topologies where one-step scalar moments effectively collapse.

\section{Empirical application: peer effects and peer preferences in NetHealth}
\label{sec:nethealth}

This section illustrates how the geometry-based instrument construction can be implemented in a
longitudinal social network setting and how conclusions about peer effects depend on the
\emph{peer-exposure aggregator} that disciplines what counts as a salient peer.
The goal is not to claim a single ``true'' social norm, but to show how the same dataset can support
very different peer-effect estimates depending on whether exposure is (i) mean-like (CES with small curvature),
(ii) attention-to-extremes (smooth-max), or (iii) rank-based (quantile/median norms).

\subsection{Data, waves, and core sample}
We use NetHealth, which combines (i) repeated network measurement and (ii) objective outcomes from wearable
devices and administrative records.
We focus on a core panel of participants with valid identifiers across the Fitbit, network, and baseline survey files.
The analysis proceeds in ``waves'' indexed by $t=1,\dots,8$.
For each wave we construct a short outcome window around the median survey date (a fixed-length window in days)
and compute wave-specific outcomes and networks; then we stack the resulting wave samples to estimate a single
(peer-effect, preference) parameter vector with cluster-robust inference clustered by individual (egoid).

\paragraph{Outcomes.}
We consider three outcomes:
(i) physical activity (Fitbit \texttt{steps}), (ii) sleep duration (Fitbit \texttt{minsasleep}), and
(iii) academic performance (term GPA from course records).
Fitbit activity reports daily steps (\texttt{steps}) and a compliance measure
(\texttt{complypercent}, percent minutes wearing/using the device).
Fitbit sleep reports minutes asleep (\texttt{minsasleep}).
Wave-level outcomes are aggregated over the window (subject to a minimum number of valid days), and we include
mean Fitbit compliance over the same window as a control (for GPA, compliance is set to 100 by construction).

\paragraph{Controls.}
The baseline control vector $X$ includes an intercept, a male indicator (from the BasicSurvey variable
\texttt{gender\_1}), and mean Fitbit compliance in the analysis window (average \texttt{complypercent});
for GPA, compliance is set to 100 for all observations.

\subsubsection*{Network construction}
For each wave $t$, we construct an interaction matrix $G_t$ from the observed friendship links.
We treat ties as undirected in the baseline construction and row-normalize $G_t$ for non-isolates.
Because both one-step and geometry-based instruments rely on propagation through the network, isolates contain
no peer-exposure information; in the baseline empirical pipeline we drop isolates within wave before stacking.
(We report sensitivity to alternative conventions such as retaining isolates with zero exposure, using reciprocal
ties only, or alternative weighting schemes.)

\subsection{Empirical model and peer exposure}
Let $y_{it}$ denote the outcome for individual $i$ in wave $t$ and $x_{it}$ the controls.
The empirical specification is a linear peer-effects equation with an endogenous peer exposure index
$S_{it}(\theta)$ disciplined by a structural aggregator:
\begin{equation}\label{eq:nethealth_struct}
y_{it} \;=\; x_{it}'\gamma \;+\; \lambda\, S_{it}(\theta) \;+\; \varepsilon_{it},
\qquad \E[\varepsilon_{it}\mid x_{it},G_t,Z_{it}]=0,
\end{equation}
where $\lambda$ is the peer-effect coefficient of interest and $\theta$ is the peer-preference / salience
parameter governing the exposure aggregator.

\paragraph{Aggregator menu (peer preferences).}
We compare three families of exposure indices $S_{it}(\theta)$:
\begin{enumerate}
\item \textbf{CES norm (curvature $\beta$):}
\[
S_{it}^{\mathrm{CES}}(\beta)
=
\Big(\sum_{j\neq i} g_{ij,t}\, y_{jt}^{\beta}\Big)^{1/\beta},
\]
which interpolates between mean-like and extreme-sensitive norms as $\beta$ varies.

\item \textbf{Smooth-max / attention (salience $\kappa$):}
\[
S_{it}^{\mathrm{SM}}(\kappa)
=
\kappa^{-1}\log\sum_{j\neq i} g_{ij,t}\exp(\kappa\, y_{jt}),
\]
implemented in \emph{scaled} form so the exposure remains in the units of $y$ (so $\lambda$ is interpretable
on the original outcome scale).

\item \textbf{Quantile / rank norm (quantile $q$):}
\[
S_{it}^{Q}(q)
=
\arg\min_{s}\sum_{j\neq i} g_{ij,t}\,\rho_q(y_{jt}-s),
\]
which corresponds to the weighted peer $q$-quantile (median when $q=0.5$; upper-quartile ``aspirational norm''
when $q=0.75$).
\end{enumerate}

\subsubsection*{Instruments: BRUZ vs GEO (geometry-augmented)}
Endogeneity arises because $S_{it}(\theta)$ is constructed from peers' outcomes and therefore inherits peers'
unobservables (reflection) and potentially sorting/common shocks.
We implement two instrument menus:

\paragraph{BRUZ (one-step scalar moments).}
Let $\hat y_{it}=m(x_{it})$ denote an exogenous predictor of $y_{it}$ constructed from observables.
BRUZ uses excluded variation from one-step aggregator summaries evaluated at $\hat y$ and numerical derivatives
w.r.t.\ the preference parameter:
\[
Z^{\mathrm{BRUZ}}_{it}(\theta)
=
\big[x_{it},\ \Phi(\hat y_t;G_t,\theta),\ \partial_\theta \Phi(\hat y_t;G_t,\theta)\big].
\]
This is the natural analogue of BDF/BRUZ instruments for general peer norms. 

\paragraph{GEO (multi-step geometry-induced instruments).}
Geometry augments BRUZ by adding excluded variation generated from the model-implied influence operator
$P_t(\theta;\hat y)$ (the normalized Jacobian weights of the exposure mapping evaluated at $\hat y$),
and then propagating covariates along multi-step influence paths (and optionally shells/torsion):
\[
Z^{\mathrm{GEO}}_{it}(\theta)
=
\big[Z^{\mathrm{BRUZ}}_{it}(\theta),\ (P_t^2(\theta;\hat y)X_t)_{it},\ \ldots \big].
\]
Intuitively, $P^kX$ replaces ``friends-of-friends'' with ``influencers-of-influencers'' under the
aggregator-consistent influence geometry, expanding the excluded $\sigma$-field beyond one-step summaries. 

\paragraph{Exogenous predictor and cross-fitting.}
Following best practice for generated instruments, we construct $\hat y=m(X)$ using a pooled cross-fitted
linear predictor (K-fold sample splitting), so that an observation's own realized outcome does not leak into
its own instruments through estimation of $m(\cdot)$.

\subsubsection*{Estimation and inference}
For each aggregator family, we estimate $(\gamma,\lambda,\theta)$ by a profile-IV procedure over a grid of the
preference parameter (e.g.\ $\beta$ for CES, $\kappa$ for smooth-max, and fixed $q$ for quantiles), and we report
cluster-robust standard errors clustered by egoid. In practice, BRUZ and GEO can select different profile minimizers,
so we report $\hat\theta^{\mathrm{BRUZ}}$ and $\hat\theta^{\mathrm{GEO}}$ separately along with $\hat\lambda$ and its
standard error. 

\subsection{Results: peer effects depend on the aggregator}
Table~\ref{tab:nethealth_allagg} reports stacked IV estimates for each outcome (steps, sleep, GPA) under each
aggregator (CES, smooth-max, quantile norms), comparing BRUZ and GEO.

\paragraph{Steps.}
Under CES and smooth-max exposure, both menus deliver negative peer-effect estimates, with GEO generally producing
smaller standard errors; the preferred preference parameter differs across menus (e.g.\ CES curvature differs between
BRUZ and GEO). Quantile norms yield noisier and sign-sensitive estimates, consistent with rank-based exposure being a
different behavioral object than mean-like or attention-to-extremes exposure. 

\paragraph{Sleep.}
CES exposure yields near-zero peer effects. Under attention and upper-quantile norms, GEO can select substantially
different preference parameters than BRUZ and can deliver positive point estimates, though precision varies by
aggregator and outcome. This pattern is consistent with the idea that sleep-related peer influence may operate through
salient or aspirational peers rather than average peers, but it also highlights that identification of preferences can be
fragile in sparse networks and motivates the geometry diagnostics discussed in the implementation notes. 

\paragraph{GPA.}
GPA effects are small under CES, while smooth-max and quantile norms can deliver different point estimates and different
preferred preference parameters under BRUZ vs GEO. Given the thinner usable sample for GPA (relative to Fitbit outcomes),
we treat GPA primarily as a mechanism and robustness outcome rather than a headline estimate. 

\paragraph{Takeaway for identification.}
The central empirical lesson is that ``peer effects'' are not a single estimand independent of modeling choices:
the mapping from peers' outcomes to exposure (mean-like vs attention vs rank norms) is itself a behavioral primitive,
and the geometry-based instrument menu provides a disciplined way to generate excluded variation that is consistent with
that primitive. The Monte Carlo evidence suggests that GEO can materially strengthen first-stage diagnostics in settings where
one-step moments are weak, and this application provides a practical template for applying the same logic to real network data.

\begin{table}[!htbp]\centering
\caption{NetHealth stacked IV results across outcomes and aggregators (BRUZ vs GEO).}\label{tab:nethealth_allagg}
\begin{threeparttable}
\begin{tabular}{llccrrrr}
\toprule
Outcome & Aggregator & Param(BRUZ) & Param(GEO) & $\hat\lambda$ (BRUZ) & se & $\hat\lambda$ (GEO) & se \\
\midrule
Steps& CES & 0.804 & 0.177 & -0.0937 & 0.0282 & -0.0636 & 0.0200 \\
Steps& SmoothMax & 0.050 & 0.050 & -0.0144 & 0.0036 & -0.0130 & 0.0035 \\
Steps& Quantile(q=0.50) & 0.50 & 0.50 & -0.0571 & 0.1017 & 0.0743 & 0.0808 \\
Steps& Quantile(q=0.75) & 0.75 & 0.75 & -0.0124 & 0.1111 & 0.1203 & 0.0754 \\
Sleep& CES & 0.050 & 0.050 & -0.0065 & 0.0123 & -0.0066 & 0.0122 \\
Sleep& SmoothMax & 10.000 & 5.150 & -0.0246 & 0.2071 & 0.1305 & 0.1197 \\
Sleep& Quantile(q=0.50) & 0.50 & 0.50 & -0.5520 & 0.3818 & -0.1309 & 0.1614 \\
Sleep& Quantile(q=0.75) & 0.75 & 0.75 & 0.0084 & 0.4649 & 0.3396 & 0.1646 \\
GPA & CES & 2.400 & 2.400 & 0.0446 & 0.0389 & 0.0464 & 0.0392 \\
GPA & SmoothMax & 0.050 & 0.800 & 0.0084 & 0.0065 & 0.1257 & 0.0851 \\
GPA & Quantile(q=0.50) & 0.50 & 0.50 & -0.0137 & 0.3647 & 0.1427 & 0.1530 \\
GPA & Quantile(q=0.75) & 0.75 & 0.75 & -0.0267 & 0.3144 & 0.0260 & 0.1924 \\
\bottomrule
\end{tabular}
\begin{tablenotes}\footnotesize
\item Notes: Each row pools (stacks) NetHealth waves for the stated outcome and reports profile-IV estimates with cluster-robust standard errors clustered by egoid. The control vector $X$ includes an intercept, a male indicator (from the BasicSurvey variable \texttt{gender\_1}), and mean Fitbit compliance in the analysis window (average \texttt{complypercent}); for GPA, compliance is set to 100 for all observations. For CES, the profile parameter is the curvature $\beta$. For SmoothMax, the profile parameter is the attention parameter $\kappa$ in $\kappa^{-1}\log\sum_j g_{ij}\exp(\kappa a_j)$ (implemented in scaled form so the exposure remains in the outcome's units). For Quantile, the profile parameter is the quantile level $q$ defining the weighted peer quantile norm. BRUZ and GEO can select different profile minimizers; therefore the table reports Param(BRUZ) and Param(GEO) separately.
\end{tablenotes}
\end{threeparttable}
\end{table}

\section{Implementation Notes}
\label{sec:implementation}

This section provides a replicable recipe for taking the model to data.
The goal is not to prescribe a single “correct” pipeline, but to lay out default steps that mirror standard applied practice
while making the construction of geometry-induced instruments transparent.

\paragraph{Objects.}
Fix a group index $s$ (e.g.\ school) and let $G_s$ denote the observed network (or any baseline connectivity object).
The empirical outcome equation features an endogenous peer-exposure index
$E_{is}(\beta;\cdot)$ (e.g.\ a CES norm, smooth-max, quantile norm), with peer-preference parameter(s) $\beta$.
Instruments are generated from the induced transport (influence) matrix
$P_s(\beta;\hat y)$, evaluated at a predetermined proxy $\hat y=m(X)$, and then transformed into
multi-step, shell-based, and torsion-weighted objects. Let $Z_{is}(\beta)$ collect all such instruments.
\paragraph{Algorithm 1 (Replication checklist).}
\begin{enumerate}
  \item \textbf{Choose a $\beta$ strategy.} Select one of:
  \begin{enumerate}
    \item \emph{(A) Joint estimation (GMM):} estimate $\beta$ together with the peer-effect parameter(s).
    \item \emph{(B) Grid sensitivity:} fix a grid $\{\beta_b\}_{b=1}^B$ and report sensitivity.
  \end{enumerate}

  \item \textbf{Cross-fit the predetermined proxy $\hat y=m(X)$.}
  Split observations into $K$ folds. For each fold $k$, estimate $m(\cdot)$ on the other folds and predict
  $\hat y^{(-k)}$ on fold $k$.

  \item \textbf{For each fold $k$ and each $\beta$ (estimated or fixed), construct transport and instruments.}
  Within each group $s$:
  \begin{enumerate}
    \item Compute the induced transport $P_s(\beta;\hat y^{(-k)})$.
    \item Build instrument families $Z_{is}(\beta)$ from $P_s(\beta;\hat y^{(-k)})$:
    multi-step transport, geodesic shells, and/or torsion-weighted objects.
    \item Construct the endogenous exposure $E_{is}(\beta;\cdot)$ entering the outcome equation.
  \end{enumerate}

  \item \textbf{Estimate the outcome equation.}
  \begin{enumerate}
    \item If the outcome equation is linear in $E_{is}$, run IV/2SLS of $y_{is}$ on $E_{is}(\beta;\cdot)$ and controls
    using $Z_{is}(\beta)$.
    \item If $\beta$ is estimated jointly or moments are nonlinear, run GMM with moments
    $\E\!\left[ Z_{is}(\beta)\,\varepsilon_{is}(\theta)\right]=0$ (optionally augmented with one-step BRUZ objects).
  \end{enumerate}

  \item \textbf{How $\beta$ is estimated in route (A) (one practical implementation).}
  \begin{enumerate}
    \item Initialize $\beta^{(0)}$ (and other parameters) and iterate:
    \item Given $\beta^{(t)}$, recompute $P_s(\beta^{(t)};\hat y^{(-k)})$, instruments $Z_{is}(\beta^{(t)})$,
    and exposures $E_{is}(\beta^{(t)};\cdot)$.
    \item Update $\theta^{(t+1)}=(\beta^{(t+1)},\ldots)$ by minimizing the GMM criterion based on the stacked moments.
  \end{enumerate}
  (In practice, a small number of outer iterations is often sufficient; convergence can be monitored via the GMM objective
  and stability of $\beta$.)

  \item \textbf{Inference and reporting.}
  Cluster at the group level $s$ at minimum. Report first-stage strength (partial $R^2$, $F$ statistics) and
  weak-IV robust tests when feasible. Under route (B), plot estimates and diagnostics over $\beta_b$.
\end{enumerate}

\begin{itemize}
  \item \textbf{What you actually estimate.}
  In applications you estimate a structural peer-effects equation (parametric or semiparametric) in which the endogenous regressor
  is a peer exposure index $E_{is}(\beta;\cdot)$, and you instrument this exposure using $Z_{is}(\beta)$ generated from
  $P_s(\beta;\hat y)$ and (optionally) shells and torsion.
  Concretely, you run either:
  \begin{enumerate}
    \item IV/2SLS, if the outcome equation is linear in the peer exposure index, or
    \item GMM, if $\beta$ is estimated jointly with the peer-effect coefficient(s), or if you impose nonlinear moment restrictions.
  \end{enumerate}
  Relative to one-step approaches (e.g.\ BRUZ-type objects), the incremental step is that instruments are built from the
  induced transport $P_s(\beta;\hat y)$ and its geometry (multi-step propagation, shells, torsion).

  \item \textbf{Choosing or estimating $\beta$ (peer preference).}
  Two practical approaches cover most empirical workflows.

  \emph{(A) Structural route: estimate $\beta$ (BRUZ-style, strengthened).}
  Treat $\beta$ as structural and estimate it by GMM jointly with the peer-effect parameter(s).
  Practically, define moments of the form $\E[z_{is}(\theta)\,\varepsilon_{is}(\theta)]=0$ where $z_{is}(\theta)$ includes
  one-step objects (e.g.\ $\Phi_i(\hat y;\beta)$ and $\partial_\beta \Phi_i(\hat y;\beta)$) and augment them with the
  geometry instruments $Z_{is}(\beta)$ to strengthen identification and improve precision.

  \emph{(B) Transparent route: sensitivity over a grid.}
  Fix a grid $\beta\in\{\beta_1,\dots,\beta_B\}$ (e.g.\ from strongly min-like to strongly max-like),
  compute instruments for each $\beta_b$, and report how the estimated peer effect and first-stage diagnostics vary with $\beta$.
  When $\beta$ is weakly identified, this robustness display is often more persuasive than a single noisy point estimate.

  \item \textbf{Constructing the predetermined proxy $\hat y=m(X)$.}
  The role of $\hat y$ is to provide a predetermined proxy for outcomes so that the induced influence weights
  $P_s(\beta;\hat y)$ are functions of observables only.

  \emph{Baseline (simple, replicable).} Use an OLS predictor
  \[
    \hat y_{is} = x_{is}'\hat\pi,
  \]
  where $\hat\pi$ is estimated on the relevant sample (or on a training subsample; see cross-fitting below).

  \emph{Flexible option (if needed).} Use richer $m(x)$ (lasso, random forest, boosting, etc.) if you have many covariates
  and care about predictive power. For validity, the key requirement is that $\hat y$ is measurable with respect to observables;
  with cross-fitting, you avoid leaking an observation's own outcome into its own instrument through the estimation of $m(\cdot)$.

  \emph{Good practice.} Include the same controls $x_{is}$ in the structural equation that you used in $m(x)$ (or a superset),
  so any direct effect of observables is absorbed and instruments primarily shift outcomes through peers.

  \item \textbf{Generated instruments and inference (why sample splitting helps).}
  Instruments are generated because $\hat y$ is estimated and $P_s(\beta;\hat y)$ depends on $\hat y$.
  This is not problematic in principle, but it can complicate inference if $m(\cdot)$ is very flexible.

  \emph{Default solution: sample splitting / cross-fitting.}
  \begin{enumerate}
    \item Split the sample into $K$ folds (often $K=2$ is sufficient in practice).
    \item For fold $k$, estimate $\hat y=m(X)$ using all observations \emph{not} in fold $k$.
    \item Compute instruments $Z_{is}(\beta)$ for observations in fold $k$ using the out-of-fold $\hat y$.
    \item Stack folds and run IV/GMM using the cross-fitted instruments.
  \end{enumerate}
  Intuition: the observation’s own outcome does not enter its own instruments through the training of $m(\cdot)$.
  This places the asymptotics closer to the textbook “instruments are predetermined” case and matches standard practice
  with generated regressors/instruments in modern applied work.

  \item \textbf{Computing geodesics and torsion (algorithmic details).}
  All geometry objects are computed \emph{within each group} $s$ using the induced influence matrix $P_s(\beta;\hat y)$.

  \emph{Geodesics (effective-distance shells).}
  \begin{enumerate}
    \item Define directed edge lengths
    \[
      \ell_{ij,s}(\beta;\hat y) = -\log\!\big(P_{ij,s}(\beta;\hat y)+\varepsilon_0\big),
    \]
    where $\varepsilon_0$ is small (e.g.\ $10^{-8}$) to avoid $\log(0)$.
    \item For each node $i$, run Dijkstra’s algorithm (or any shortest-path routine) on the weighted directed graph
    to compute effective distances $d_{\beta,s}(i,j)$ to all $j$ in group $s$.
    \item Form shells $\mathcal{S}_{is}(h)=\{j:\ d_{\beta,s}(i,j)\in(h-1,h]\}$ and compute shell instruments
    \[
      Z^{\mathrm{shell}}_{is,h}(\beta)=\sum_{j\in\mathcal{S}_{is}(h)} x_{js}.
    \]
  \end{enumerate}
  \emph{Defaults and notes.} (i) Start shells at $h\ge 2$ to preserve exclusion relative to direct-peer terms; (ii) in large groups,
  cap the maximum distance and/or the number of visited nodes for speed; (iii) typical empirical implementations consider
  $h=2,\dots,H$ for modest $H$ (e.g.\ $H\in[4,8]$), reporting sensitivity to $H$.

  \emph{Torsion (path non-redundancy).}
  The raw wedge definition uses triples $(i,j,k)$ and can be computationally heavy in large groups.
  Two practical implementations are:
  \begin{enumerate}
    \item \emph{Wedge-based (denser networks).}
    For each $i$, iterate over $j$ with $P_{ij,s}>0$ and $k$ with $P_{jk,s}>0$, compute
    $\tau_{i,j,k,s}(\beta)=\big|P_{ik,s}-P_{ij,s}P_{jk,s}\big|$, and accumulate the torsion-weighted exposure
    \[
      Z^{\tau}_{is}(\beta)=\sum_{j,k} P_{ij,s}P_{jk,s}\tau_{i,j,k,s}(\beta)\,x_{ks}.
    \]
    \item \emph{Competing-short-paths (sparser networks).}
    For each $(i,k)$, compare the direct weight $P_{ik,s}$ to the strongest (or shortest-distance) two-step route through some $j$.
    Use the discrepancy as a torsion weight and construct an exposure that emphasizes pairs where two-step propagation
    is not well summarized by one-step propagation.
  \end{enumerate}
  Intuition: torsion is large exactly when higher-order propagation contains independent variation, which is where
  multi-step instruments are most likely to strengthen the first stage.

  \item \textbf{Dependence and standard errors.}
  Peer-effects data are typically clustered by group $s$ (e.g.\ common shocks, policies, cohort effects) and can also exhibit
  within-group network dependence.

  \emph{Minimum: cluster by group.}
  Use cluster-robust standard errors at the group level ($s$). This is the default in most peer-effects applications.

  \emph{If groups are large and network dependence is strong.}
  Consider variance estimators that allow correlation to decay with network distance, or cluster by network communities within groups.
  A practical justification is to show residual correlation decays with effective distance $d_{\beta,s}$ (computed above).

  \emph{Reporting.}
  Regardless of the SE choice, report first-stage strength diagnostics (e.g.\ partial $R^2$, first-stage $F$ statistics, and
  weak-instrument-robust tests where feasible). The empirical promise of geometry-induced instruments is relevance:
  they align excluded variation with the model’s induced propagation operator $P_s(\beta;\hat y)$ rather than with adjacency alone.
\end{itemize}
\section{Conclusion}\label{sec:conclusion}

This paper offers a unifying perspective on peer-effects identification: propagation is governed by a transport operator induced by the peer aggregator.
The linear-in-means model is the constant-transport benchmark ($P=G$), which recovers the familiar $G^kX$ instrument families and clarifies their well-known fragilities under transitivity, regularity, and path redundancy.
By contrast, peer-preference aggregators---such as CES norms and their extensions---generate a state-dependent Jacobian transport operator.

Its powers $P^k$, along with geometry-induced objects (shells/geodesics and torsion), produce new excluded-variation families tailored to how influence actually flows under the maintained aggregator.

The main implication is practical: identification can be strengthened by moving from one-step summaries to \emph{multi-step, influence-consistent} propagation.
Because geometry-based instruments exploit marginal influence heterogeneity and non-redundant paths, they can remain informative precisely in settings where one-step scalar moments (e.g., BRUZ-type summaries) or algebraic powers of $G$ become nearly collinear and weak.
More broadly, the framework makes explicit that ``peer effects'' depend on the behavioral primitive used to map peer outcomes into exposure; different aggregators induce different transports and therefore different estimands.
This transport view provides a constructive route to diagnosis (via profile curvature and first-stage strength) and to design (via instrument menus that match the implied influence geometry), and it suggests natural extensions to alternative aggregators such as attention (smooth-max) and rank norms (quantiles).

\newpage

\bibliographystyle{aer} 
\bibliography{bibALL}   
\appendix

\section{Appendix: additional aggregator class and technical bounds}
\label{app:aggregators_and_bounds}

This appendix adds an additional peer-aggregator category (quantile/median norms) and provides easy-to-check bounds
that verify Assumption~\ref{ass:lipschitz_main} for the workhorse CES and smooth-max cases.

\subsection{Loss-minimizer norms (median/quantiles/trimmed means)}
\label{app:loss_minimizer}

A natural economic primitive is that a peer ``norm'' is the action level that minimizes expected social disapproval,
modeled as a loss function. Fix a convex loss $\ell:\mathbb{R}\to\mathbb{R}_+$ and define
\begin{equation}\label{eq:loss_norm}
\Phi_i^{\text{LM}}(a_{-i};G)
\;\in\;
\arg\min_{m\in M}\ \sum_{j\neq i} g_{ij}\,\ell(m-a_j),
\end{equation}
where $M\subset\mathbb{R}$ is compact.

\paragraph{Examples.}
\begin{itemize}
\item If $\ell(u)=u^2$, then $\Phi_i^{\text{LM}}$ is the weighted mean $(Ga)_i$.
\item If $\ell(u)=|u|$, then $\Phi_i^{\text{LM}}$ is the weighted median.
\item If $\ell(u)=u(\tau-\mathbf{1}\{u<0\})$ (check loss), then $\Phi_i^{\text{LM}}$ is the weighted $\tau$-quantile.
\end{itemize}

\paragraph{Why this matters for nonparametric peer effects.}
Loss-minimizer norms allow peer exposure to depend on order statistics (median/quantiles) rather than averages.
This is a natural route to ``nonlinear'' and nonparametric peer effects that remain economically interpretable.

\paragraph{Existence with set-valued aggregators.}
When $\ell$ is convex, the minimizer set in \eqref{eq:loss_norm} is nonempty and convex.
Thus $\Phi_i^{\text{LM}}$ can be treated as a set-valued correspondence.
The existence theorem (Theorem~\ref{thm:existence_norm_game}) continues to hold under the standard set-valued extension:
assume $\Phi_i^{\text{LM}}$ is upper hemicontinuous with nonempty convex values, and $U_i$ is concave in $a_i$.
Then Kakutani yields equilibrium.

\paragraph{Replicable verification (upper hemicontinuity).}
Define
\[
Q_i(m;a_{-i})=\sum_{j\neq i} g_{ij}\,\ell(m-a_j).
\]
If $\ell$ is continuous, then $Q_i$ is continuous in $(m,a_{-i})$.
If $M$ is compact, Berge's maximum theorem implies the argmin correspondence is upper hemicontinuous and nonempty.
Convexity of $\ell$ implies convexity of the argmin set.
Hence the conditions needed by Kakutani are satisfied.

\subsection{Lipschitz bounds for CES and smooth-max norms}
\label{app:lipschitz_bounds}

This subsection gives sufficient conditions under which the aggregator classes satisfy a Lipschitz bound of the form
\eqref{eq:lipschitz_phi}.

\subsubsection{Smooth-max (log-sum-exp) bound}
Fix $\kappa>0$ and define $\Phi_i^{\text{SM}}$ as in \eqref{eq:smoothmax_main}.
For any vectors $a,a'$,
\begin{align}
\left|\Phi_i^{\text{SM}}(a_{-i})-\Phi_i^{\text{SM}}(a'_{-i})\right|
&=
\left|\frac{1}{\kappa}\log\frac{\sum_{j\neq i} g_{ij}e^{\kappa a_j}}{\sum_{j\neq i} g_{ij}e^{\kappa a'_j}}\right|
\;\le\;
\|a-a'\|_\infty. \label{eq:sm_lipschitz}
\end{align}
Thus one can take $L_\Phi=1$ under the $\ell_\infty$ norm.
(Proof sketch: log-sum-exp is 1-Lipschitz in $\ell_\infty$ because its gradient is a probability vector.)

Combining \eqref{eq:sm_lipschitz} with the logit mapping in \eqref{eq:bd_prob_main} yields a contraction condition:
since $\sup_t\Lambda'(t)\le 1/4$, the composite map has Lipschitz constant at most $|J|/4$.
Therefore uniqueness holds whenever $|J|<4$ (or more generally $|J|L_\Phi<4$).

\subsubsection{CES bound on a compact, positive domain}
Fix $\beta\neq 0$ and suppose actions lie in $A_i\subset[\underline{a},\bar a]$ with $0<\underline{a}<\bar a<\infty$.
Then $\Phi_i^{\text{CES}}(\cdot;\beta)$ is continuously differentiable in $a_{-i}$.
By the mean value theorem,
\begin{equation}\label{eq:ces_mvt}
\left|\Phi_i^{\text{CES}}(a_{-i})-\Phi_i^{\text{CES}}(a'_{-i})\right|
\le
\sum_{j\neq i}\sup_{\tilde a}\left|\frac{\partial \Phi_i^{\text{CES}}(\tilde a_{-i})}{\partial a_j}\right|\cdot |a_j-a'_j|.
\end{equation}
A direct differentiation gives, for $j\neq i$,
\begin{equation}\label{eq:ces_deriv_app}
\frac{\partial \Phi_i^{\text{CES}}(a_{-i})}{\partial a_j}
=
g_{ij}
\left(\sum_{k\neq i} g_{ik}a_k^\beta\right)^{\frac{1}{\beta}-1}
a_j^{\beta-1}.
\end{equation}
On $[\underline{a},\bar a]$, the term $a_j^{\beta-1}$ is bounded and the bracketed term is bounded away from $0$.
Hence the derivative is uniformly bounded:
\[
\sup_{\tilde a\in [\underline{a},\bar a]^n}\left|\frac{\partial \Phi_i^{\text{CES}}(\tilde a_{-i})}{\partial a_j}\right|
\le C(\underline{a},\bar a,\beta)\, g_{ij},
\]
for some finite constant $C(\underline{a},\bar a,\beta)$.
Plugging into \eqref{eq:ces_mvt} yields a Lipschitz bound of the form
\[
\left|\Phi_i^{\text{CES}}(a_{-i})-\Phi_i^{\text{CES}}(a'_{-i})\right|
\le
C(\underline{a},\bar a,\beta)\,\|a-a'\|_\infty,
\]
because $\sum_{j\neq i} g_{ij}=1$.
Thus one may take $L_\Phi=C(\underline{a},\bar a,\beta)$ under $\ell_\infty$.

\paragraph{Implication for uniqueness.}
If the best-response sensitivity \eqref{eq:lipschitz_b} holds with constant $L_b$,
then Theorem~\ref{thm:unique_contraction_main} implies uniqueness whenever
$L_b\,C(\underline{a},\bar a,\beta)<1$.

\subsection{Replicable mapping to Brock--Durlauf probability fixed points}
\label{app:bd_mapping}

For binary actions, pure equilibria may fail to exist, but probability equilibria are natural and standard.
Let $p_i\in[0,1]$ denote the probability that $i$ chooses action $1$.
Assume the latent utility difference is
\[
\Delta u_i = x_i'\gamma + J\,\Phi_i(p_{-i};G,\theta) + \eta_i,
\]
with i.i.d.\ logistic shocks $\eta_i$.
Then the choice probability is
\[
p_i = \Pr(\Delta u_i\ge 0) = \Lambda\!\Big(x_i'\gamma + J\,\Phi_i(p_{-i};G,\theta)\Big),
\]
which is exactly \eqref{eq:bd_prob_main}.
Define the mapping $T:[0,1]^n\to [0,1]^n$ by $(T(p))_i=\Lambda(x_i'\gamma+J\Phi_i(p_{-i}))$.
Continuity of $\Phi$ implies continuity of $T$, so Brouwer yields existence of a fixed point $p^\star$.
Uniqueness follows by the contraction logic using $\sup\Lambda'\le 1/4$ and a Lipschitz bound for $\Phi$
(as shown for smooth-max in \eqref{eq:sm_lipschitz}).

\end{document}